\documentclass[a4paper, USenglish, cleveref, nameinlink, autoref]{lipics-v2021}

\usepackage{graphicx}
\usepackage[utf8]{inputenc}
\usepackage{xspace}
\usepackage{todonotes}
\usepackage{mathtools}
\usepackage{orcidlink}
\usepackage{booktabs}
\usepackage[inline]{enumitem}
\usepackage{subcaption}
\usepackage{cite}
\usepackage{amsthm}
\usepackage{comment}
\usepackage[absolute,overlay]{textpos}

\crefname{appendix}{Appendix}{Appendices}
\crefname{theorem}{Thm.}{Thms.}
\crefname{theorem}{Theorem}{Theorems}
\crefname{corollary}{Cor.}{Cors.}
\crefname{corollary}{Corollary}{Corollaries}

\graphicspath{{figures/}}

\definecolor{defblue}{rgb}{0.121,0.47,0.705}
\DeclareTextFontCommand{\emph}{\color{defblue}\em}
\hypersetup{colorlinks=true,
    linkcolor=defblue,
    anchorcolor=defblue,
    citecolor=defblue,
    filecolor=defblue,
    menucolor=defblue,
    urlcolor=defblue,
    bookmarksopen=true,
    bookmarksopenlevel=2,
    bookmarksnumbered=true,
    plainpages=false
    }

\newcounter{slope}
\crefname{slope}{}{}
\crefformat{slope}{#2#1#3}
\crefmultiformat{slope}{#2#1#3}{, #2#1#3}{, #2#1#3}{, #2#1#3}
\makeatletter
\newcommand{\slopeitem}[2]{%
  \item[$#1$]%
  \refstepcounter{slope}%
  \protected@edef\@currentlabel{$#1$}%
  \label{#2}%
}
\makeatother

\title{On the \texorpdfstring{$2$}{2}-Bend Slope Number of \texorpdfstring{$1$}{1}-Planar Graphs\\
\small{The Slopebusters give the first upper bound}}

\titlerunning{On the \texorpdfstring{$2$}{2}-Bend Slope Number of \texorpdfstring{$1$}{1}-Planar Graphs}

\author{Michael A. Bekos}{University of Ioannina, Greece \and \url{https://myweb.uoi.gr/bekos/} }{bekos@uoi.gr}{https://orcid.org/0000-0002-3414-7444}{}
\author{Eleni Katsanou}{National Technical University of Athens,  Greece} {ekatsanou@mail.ntua.gr}{https://orcid.org/0000-0002-1001-1411}{}
\author{Philipp Kindermann}{Trier University, Germany \and \url{https://algo.uni-trier.de/~kindermann}} {kindermann@uni-trier.de}{https://orcid.org/0000-0001-5764-7719}{}
\author{Aikaterini Maria Ntasiou}{University of Ioannina, Greece \and \url{https://algo.math.uoi.gr/maritina/}}{m.ntasiou@uoi.gr}{https://orcid.org/0009-0006-7481-3268}{}
\author{Maria Eleni Pavlidi}{University of Ioannina, Greece \and \url{https://algo.math.uoi.gr/marialena/index.html}}{m.e.pavlidi@uoi.gr}{https://orcid.org/0009-0009-4500-0112}{}
\author{Soeren Terziadis}{TU Munich, Germany}{sterziadis@ac.tuwien.ac.at}{https://orcid.org/0000-0001-5161-3841}{}

\authorrunning{M. A. Bekos, E. Katsanou, P. Kindermann, A. M. Ntasiou, M. E. Pavlidi, S. Terziadis}

\Copyright{Michael A. Bekos, Eleni Katsanou, Philipp Kindermann, Aikaterini Maria Ntasiou, Maria Eleni Pavlidi, Soeren Terziadis}

\acknowledgements{The work was partially funded by the IKYDA 2025 funding program of the DAAD (project 57775422 ``Algorithms and Combinatorics for Drawing Graphs with Few Slopes'').}



\nolinenumbers
\hideLIPIcs

\begin{document}

\maketitle

\begin{abstract}
\begin{textblock*}{6.5cm}(16.9cm,4.95cm)
\includegraphics[width=0.35\textwidth]{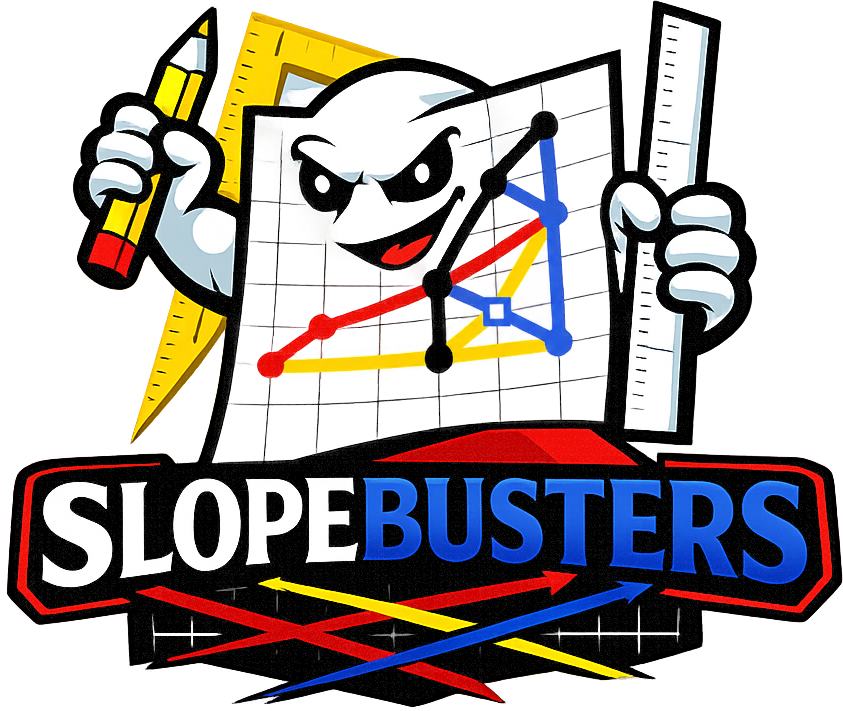}
\end{textblock*}
While drawing planar graphs with few slopes and few bends is a well-studied problem, corresponding extensions to beyond-planar graphs still remain mostly unexplored. 
Motivated by this observation, in this work, we provide bounds on the slope number of biconnected 1-planar graphs when two bends are allowed along each edge. Our contribution is an incremental drawing algorithm that produces 2-bend 1-planar drawings of biconnected 1-plane graphs with maximum degree $\Delta$ using \textit{any} prescribed set of $\Delta$ pairwise distinct slopes.
\end{abstract}

\section{Introduction}
The problem of computing drawings of graphs in which the edges use only a limited number of slopes has received considerable attention in Graph Drawing, due to both its theoretical significance and practical relevance in areas such as network visualization and VLSI design \cite{BattistaETT99,Juenger04,2013gd}.
Given a graph, algorithmically one seeks to minimize (or to compute bounds on) its \emph{slope number}, that is, the number of distinct slopes used by its edge segments, taken over all polyline drawings of the graph with a prescribed number of bends per edge. Among the first to investigate this problem were Wade and Chu \cite{WadeChu1994}, who proved that the slope number of $K_n$ in the straight-line setting is $n$. Since then, the problem has been extensively investigated with several results known for graphs of bounded degree~\cite{DujmovicSW07, KeszeghPPT06,KeszeghPPT08, MukkamalaP11, MukkamalaSzegedy2009, PachP06}.

For planar input graphs, one usually seeks to find corresponding crossing-free drawings.
Several results have established tight or near-tight bounds on the planar slope number, depending on the maximum degree of the graph and on the number of bends per edge allowed~\cite{BekosGDLM22, KlawitterM22, KlawitterZ21,KlawitterZ23,ChaplickLGLM21,ChaplickLGLM24, DiGiacomoLM2014,GiacomoLM18, GiacomoLM16,GiacomoLM20, DujmovicESW07, JelinekJKLTV09,JelinekJKLTV13, KPP2010GD,KPP2013, KnauerMW12,KnauerMW14, LenhartLMN23, BrucknerKM19,BrucknerKM22,KnauerW16}. Additionally, trade-offs between the number of slopes, the number of bends, and drawing area have been proposed, yielding efficient  algorithms to produce corresponding drawings~\cite{BekosKKP26, BiedlKant1998, BekosGKK14,BekosGKK15,Tamassia1987}. By now, the required number of slopes in combination with other parameters in the planar setting is fairly well~understood.

In the non-planar setting, Pach and Pálvölgyi~\cite{PachP06} showed that graphs of maximum degree $5$ can have arbitrarily large slope number, while if one bend per edge is allowed
Dujmovi\'c, Suderman, and Wood~\cite{DujmovicSW07} proved that every graph of maximum degree~$\Delta$ admits a drawing using $\Delta+1$ slopes. However, if the input graph belongs to a particular graph class, then the output drawing does not necessarily respect the defining properties of the class. A prominent example in this regard is the class of 1-planar graphs, which have a drawing where each edge is crossed at~most once. 
Introduced by Ringel~\cite{Ringel1965}, 1-planar graphs 
have been the subject of extensive research~\cite{DidimoLM19,KobourovLM17}. However, their slope number has been studied to a limited extent. To the best of our knowledge, the only non-trivial bounds on their slope numbers concern subclasses. Kindermann, Montecchiani, Schlipf, and Schulz~\cite{KindermannMSS18,KindermannMSS21} showed that triconnected cubic 1-planar graphs admit 1-planar drawings with one bend per edge using four slopes.~Di~Giacomo, Liotta, and Montecchiani~\cite{GiacomoLM14,GiacomoLM15}~showed that outer 1-planar graphs of degree $\Delta$ admit outer 1-planar drawings with~one~bend per edge using $O(\Delta)$ slopes. Argyriou, Cornelsen, F{\"{o}}rster, Kaufmann, N{\"{o}}llenburg, Okamoto, Raftopoulou, and Wolff~\cite{ArgyriouCF0NORW18} presented an algorithm to draw 1-plane graphs with maximum degree four orthogonally (i.e., with two slopes) and with at most three bends per edge.
Most existing results, however, in non-planar settings focus either on dense graphs, such as complete and complete~bipartite graphs~\cite{WadeChu1994,DujmovicSW07}, or on graphs of small degree~\cite{KeszeghPPT06,KeszeghPPT08,MukkamalaP11}. As a consequence, sparse yet structured non-planar graph families remain largely unexplored.

This lack of results highlights a gap in the literature: \textit{while the slope number problem~is well studied in the planar case, its extension to beyond-planar graph classes is still in its infancy}. Understanding how different crossing restrictions affect the number of slopes, the required number of bends per edge, and the drawing area is a challenging and largely unexplored research line, which we seek to investigate. To this end, we leverage a central technique for computing drawings with few slopes and few bends per edge that was introduced in~\cite{AngeliniBLM19} by extending previous techniques for graphs of low degree~\cite{BiedlKant1998,BekosGKK14,BekosGKK15}. Initially, it was applied to planar graphs of degree~$\Delta$~\cite{AngeliniBLM19}, but later it was extended to produce corresponding drawings for directed planar graphs~\cite{BekosGDLM18}, for triconnected cubic $1$-planar graphs~\cite{KindermannMSS18,KindermannMSS21}, and recently to study the interplay between the $k$-bend planar slope number and the area requirements~\cite{BekosKKP26}.

\subparagraph{Our contribution.} 
We extend the aforementioned technique and adapt it to the $2$-bend setting for $1$-planar graphs of arbitrary degree.
In~\cref{sec:triconnected}, we consider $3$-connected $1$-planar graphs.
Given such a graph of maximum degree $\Delta$ and a 1-planar embedding of it, we construct a $2$-bend $1$-planar drawing using an arbitrary (prescribed) set of $\Delta$ slopes, not necessarily maintaining the given embedding.
The drawing is obtained via an incremental algorithm based on a canonical ordering~\cite{Kant96} of a planarization of the input $1$-planar graph.
In~\cref{sec:biconnected}, the result is extended to biconnected $1$-planar graphs by exploiting SPQR-tree decompositions~\cite{BattistaT90}. Our main result is as~follows.

\begin{theorem}\label{thm:main}
The $2$-bend slope number of $2$-connected $1$-planar graphs with degree $\Delta$ is at most $\Delta$.    
\end{theorem}

\section{Preliminaries}\label{sec:preliminaries}

We assume familiarity with basic Graph Drawing terms and techniques \cite{BattistaETT99, 2013gd}.
Unless stated otherwise, all graphs considered are simple and undirected. The \emph{degree} of a vertex is the number of its neighbors. A graph has \emph{maximum degree}~$\Delta$ if it contains a vertex of degree~$\Delta$ and no vertex of degree greater than~$\Delta$. A graph is \emph{connected} if every pair of vertices is joined by a path. More generally, for $k \ge 1$, a graph is \emph{$k$-connected} if the removal of any set of at most $k-1$ vertices leaves the graph connected. In particular, $2$- and $3$-connected graphs are also referred to as \emph{biconnected} and \emph{triconnected}, respectively.

A \emph{drawing} of a graph maps each vertex of the graph to a point of the Euclidean plane and each of its edges to a Jordan arc connecting its endpoints. A drawing is \emph{planar} if no two edges intersect except possibly at common endpoints. Such a drawing partitions the plane into connected regions called \emph{faces}; the unbounded one is the \emph{outer face}. A graph is \emph{planar} if it admits a planar drawing. A \emph{planar embedding} of a planar graph is an equivalence class of planar drawings that define the same set of faces and the same outer face. A \emph{plane} graph is a planar graph with a fixed planar embedding. A drawing is \emph{$k$-bend} if each of its edges is a polygonal chain composed of at most $k+1$ straight-line segments. The point where two such segments meet is called a \emph{bend}. 
A graph is called \emph{1-planar} if it admits a drawing in the plane in which each edge is crossed at most once. 
A \emph{1-planar embedding} is a combinatorial description of the faces of a 1-planar drawing.
A \emph{1-plane} graph is a graph together with a fixed such embedding.
Given a 1-plane graph, its \emph{planarization} is the plane graph obtained by replacing each crossing with a \emph{dummy} vertex; non-dummy vertices in the planarization are referred to as \emph{real} vertices.  

The \emph{slope} of a line measures its steepness and direction. 
It is defined as the ratio of the vertical change (\emph{rise}) to the horizontal change (\emph{run}) between any two points on the line. 
This equivalently corresponds to the tangent of the counterclockwise angle through which a horizontal line must be rotated to coincide with the given line. 
The \emph{horizontal} (\emph{vertical}) \emph{slope} is the slope of a line parallel (perpendicular) to the $x$-axis. 
The slope of an edge segment is the slope of the line containing it. 
Given a set of slopes $S$, a $k$-bend planar drawing is said to be \emph{on $S$} if each of its edge segments has a slope belonging to $S$. 
For each vertex $v$ of a $k$-bend planar drawing on $S$ and each slope $s\in S$, the two opposite directions determined by $s$ at $v$ define two attachment points, called the \emph{ports} of slope $s$ at $v$.
Geometrically, each port is represented by the ray emanating from~$v$ in the corresponding direction.
If $s$ is the horizontal slope, these rays are called \emph{horizontal}. 
The upward (downward) directed ports are called \emph{top} (\emph{bottom}) ports. 
We say that a port~$\pi_v$ incident to $v$ is \emph{free} if no edge incident to $v$ is drawn along $\pi_v$; otherwise,~$\pi_v$ is \emph{occupied}.  
A set of slopes $S$ is called \emph{universal} for a family of graphs $\mathcal{G}$ and a drawing style $\mathcal{D}$ if every graph $G \in \mathcal{G}$ admits a $\mathcal{D}$-style drawing on~$S$. 
W.l.o.g.\ we assume that $S$ contains the horizontal slope (otherwise, we rotate the slope set). 

Given an edge $e$ on the outer face of a $k$-bend planar drawing $\Gamma$, a \emph{cut at $e$} is a strictly \mbox{$y$-monotone} curve that
\begin{enumerate*}[label=(\roman*)]
\item starts at a point on a horizontal segment of $e$,
\item ends on a horizontal segment of another edge $e'$ of the outer face of $\Gamma$, $e' \neq e$, and
\item intersects only horizontal segments of $\Gamma$; see \cref{fig:schnyder}.
\end{enumerate*}
Such a cut allows to \emph{stretch}~$\Gamma$ horizontally by translating~all vertices and edges on one side of the cut horizontally by any distance $d>0$, thereby increasing the horizontal distance between the two parts without introducing crossings or changing~the order of the edges around any vertex. Since the stretching is affecting horizontal edge segments, the slopes of all non-horizontal segments remain unchanged, while the lengths of the horizontal segments crossed by the cut increase. So, if $\Gamma$ is on $S$ before~stretching, it remains on $S$ afterwards. When we say that we \emph{stretch an edge}, we refer precisely~to~this~operation.

\begin{figure}
    \centering
    \includegraphics[page=4]{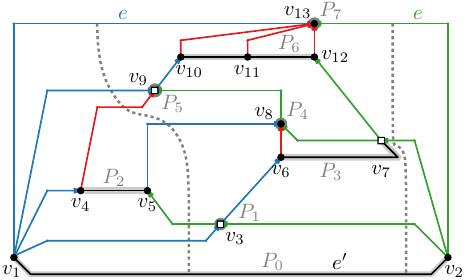}
    \caption{A canonical order and a $4$-coloring of a triconnected planarization of a $1$-plane graph. The dotted curves represent two cuts. Edge $(v_3,v_6)$ is the last edge of dummy vertex $v_3$, while $(v_1,v_3)$ is the corresponding critical edge. (Note that this is not a drawing produced by our algorithm.)}
    \label{fig:schnyder}
\end{figure}

\medskip\noindent\textbf{Canonical ordering.}
Let $G$ be a triconnected $n$-vertex plane graph and let $\Pi = (P_0,\ldots,P_m)$ be a partition of its vertex set into paths such that $P_0=\{v_1,v_2\}$, $P_m=\{v_n\}$, the edges $(v_1,v_2)$ and $(v_1,v_n)$ exist and belong to the outer face of $G$. 
For $i=0,\ldots,m$, let $G_i$ be the subgraph induced by $P_0\cup\ldots \cup P_i$ and denote by $C_i$ the \emph{contour} of $G_i$ defined as follows:  If $i=0$, then $C_0$ is the edge $(v_1,v_2)$ of $P_0$, while if $i >0$, then $C_i$ is the path from $v_1$ to $v_2$ obtained by removing $(v_1,v_2)$ from the cycle delimiting the outer face of $G_i$. 
We say that $\Pi$ is a \emph{canonical order}~\cite{FraysseixPP90,Kant96} of $G$ if for each $i=1,\ldots,m-1$ the following hold (see~\cref{fig:schnyder}): %
\begin{enumerate*}[label=(\roman*)]
\item $G_i$ is biconnected, internally triconnected and embedded with $C_i \cup \{(v_1,v_2)\}$ as its outer face; 
\item all neighbors of $P_i$ in $G_{i-1}$ are on $C_{i-1}$; 
\item $P_i$ either consists of a single vertex (called \emph{singleton}), or the degree of each of its vertices is $2$ in $G_i$ (called \emph{chain});
\item every vertex in $P_i$ has at least one neighbor in $P_j$ with $j>i$.
\end{enumerate*}
A canonical order of a triconnected planar graph can be computed in linear time~\cite{Kant96}.
We write $v_a \prec v_b$ to denote that $v_a\in P_i$, $v_b\in P_j$ and $i<j$.
For a vertex $v_h$ of $P_i$ with $1 \leq i < n$, its \emph{last neighbor} is the unique neighbor $v_g$ contained in the path $P_j$ with maximum index $j$ in $\Pi$. The corresponding edge $(v_h,v_g)$ is the \emph{last edge} of $v_h$. For a dummy vertex, the edge opposite of the last edge in the circular order of the edges around the vertex is called the \emph{critical} edge.

\medskip\noindent\textbf{4-edge coloring.} For the pair $\langle G, \Pi \rangle$, it is possible to compute a 4-edge coloring similar to the one by Schnyder~\cite{Felsner04,Schnyder90}. The edge $(v_1,v_2)$ of~$G_0$ is colored black. For $i = 1,\ldots,m$, a 4-coloring of $G_{i-1}$ is extended to one of $G_i$ as follows (see, e.g., \cref{fig:schnyder}).
First, consider the edges of $G_i$ that do not belong to $G_{i-1}$ and lie on the contour $C_i$. The first (last) such edge encountered on a traversal of $C_i$ from $v_1$ to $v_2$ is colored blue (green, respectively), while all remaining ones (i.e., those connecting vertices within $P_i$ when $P_i$ forms a chain) are colored black. 
The remaining edges of $G_i$ that do not belong to $G_{i-1}$ are colored red; these are precisely the edges incident to~$P_i$ in $G_i$ that are not part of $C_i$ (i.e., these are only present if $P_i$ is a singleton of the canonical order).
Finally, we treat all black edges as undirected, and all remaining edges as directed according to $\prec$.

\medskip\noindent\textbf{SPQR-tree decomposition.} Let  $G$ be a biconnected graph. The \emph{SPQR-tree $\mathcal{T}$} of $G$ represents a recursive decomposition of $G$ based on its split pairs~\cite{BattistaETT99}. The nodes of $\mathcal{T}$ correspond to components of this decomposition and have four types:~$S$,~$P$,~$Q$,~and~$R$;~see~\cref{fig:spqr}.

\begin{figure}
    \centering
    \includegraphics[width=0.9\linewidth]{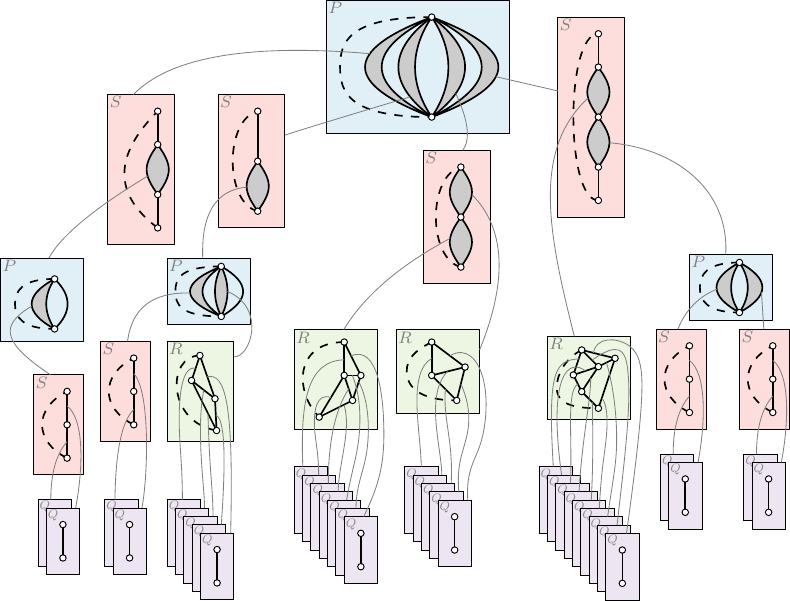}
    \caption{Illustration of an SPQR-tree decomposition.}
    \label{fig:spqr}
\end{figure}

Every split pair defines a node $\mu\in \mathcal{T}$; the vertices of the split pair are called the \emph{poles} of the node. Every node is associated with a biconnected multigraph, called the \emph{skeleton} of~$\mu$. The skeleton contains special edges, called \emph{virtual edges}, that represent connections to adjacent nodes of $\mathcal{T}$. More precisely, for every node $\nu$ adjacent to $\mu$ in $\mathcal{T}$, there is a virtual edge in the skeleton of $\mu$ corresponding to $\nu$.
Depending on the structure of its skeleton, a node $\mu$ is:   \begin{enumerate*}[label=(\roman*)]
    \item an \emph{$S$-node} if the skeleton of $\mu$ is a simple cycle,
    \item a \emph{$P$-node} if the skeleton of $\mu$ consists of parallel edges between its poles,
    \item a \emph{$Q$-node} if the skeleton of $\mu$ consists of a single edge, or
    \item an \emph{$R$-node} if the skeleton of $\mu$ is a triconnected graph.
  \end{enumerate*}
No two $S$-nodes and no two $P$-nodes are adjacent in~$\mathcal{T}$.

The \emph{pertinent graph} of a node $\mu$, denoted by $G_\mu$, is the subgraph of $G$ induced by the subtree of $\mathcal{T}$ rooted at $\mu$. It is defined recursively as follows: if $\mu$ is a $Q$-node, then $G_\mu$ consists of a single edge between its poles; otherwise, $G_\mu$ is obtained from the skeleton of $\mu$ by replacing each virtual edge corresponding to a child $\nu$ with the pertinent graph $G_\nu$.

\begin{lemma}
    \label{lem:planarization}
    Let $H$ be a $1$-plane graph and $G$ its planarization. Then we can reembed $H$ such that each edge is still crossed at most once and 
  \begin{enumerate*}[label=(\roman*)]
    \item\label{prp:bicon}no cut vertex of its planarization is a dummy vertex, and
    \item\label{prp:tricon}if~$H$ is $k$-connected with $k \in \{2,3\}$, then its planarization is also $k$-connected, and
    \item\label{prp:poles} if a pole $d$ of a separation pair $\langle d,t\rangle$ is a dummy vertex, then $(d,t)$ is not an edge of $G$, removing $\{d,t\}$ yields exactly two connected components, and $d$ has exactly one neighbor in one of these components and exactly three neighbors in the other.
  \end{enumerate*}
\end{lemma}

\begin{proof}
Property~\ref{prp:bicon} as well as Property~\ref{prp:tricon} for $k=3$ are shown by Kindermann, Montecchiani, Schlipf, and Schulz~\cite{KindermannMSS21}. Property~\ref{prp:bicon}, however, also implies Property~\ref{prp:tricon} for $k=2$. To see this, assume that the planarization of $H$ is not $2$-connected, that is, it contains a cutvertex $v$. By Property~\ref{prp:bicon}, vertex $v$ is a real vertex. Hence, $v$ is a cutvertex in the planarization of $H$ but not in $H$, which implies that it was created in the planarization process. The removal of $v$ from the planarization yields at least two connected components of it. This implies that each dummy vertex lies completely in a connected component. Therefore, $v$ is a cutvertex also in $H$; a contradiction. 

Regarding Property~\ref{prp:poles}, suppose that $d$ has exactly two neighbors in one of the connected components obtained by removing $\{d,t\}$.
Then, by flipping this component, the crossing corresponding to $d$ is eliminated in $H$. If removing $\{d,t\}$ yields more than two connected components, then we can reorder them around~$t$ to remove the crossing corresponding to~$d$.
Note that if $t$ is also such a dummy vertex, then the crossing corresponding to $t$ is also eliminated. Thus, there are exactly two connected components; since~$d$ has degree~4, one of them contains exactly three neighbors of $d$, and the other one contains exactly one neighbor of $d$.
In particular, all four neighbors of $d$ lie in these two components, so $(d,t)$ is not an edge of $G$. Applying this procedure to every such separation pair, yields a new embedding of $H$ which has the claimed property.
\end{proof}

\section{The triconnected case}
\label{sec:triconnected}

In this section, we show that \emph{any} set $S$ of $\Delta$ slopes is universal for the class of triconnected $1$-planar graphs of degree $\Delta$.
To this end, let $H$ be a triconnected $1$-plane graph and $G$ its planarization on $n$ vertices, which by \cref{lem:planarization} is triconnected.
Let~$\Pi = (P_0, \ldots, P_m)$~be a canonical order of $G$, such that $P_0=\{v_1,v_2\}$ and $P_m=\{v_n\}$ with $v_1$ and $v_2$ being real vertices, and $(v_1,v_2)$ and $(v_1,v_n)$ being edges on the outer face of~$G$. Afterwards, we describe how to adjust our approach if there is no edge between two real vertices.
Finally, let $\delta$ be the maximum run of any non-horizontal slope in~$S$, when the rise is~$1$.

We use a technique similar to the algorithm of~\cite{AngeliniBLM19}, which constructs $1$-bend drawings of planar graphs.
To this end, they recursively construct a drawing by installing the vertices following the canonical order above the already constructed drawing. To accommodate the new vertices, they make sure that every edge on the contour has a horizontal segment that can be used to stretch the drawing.
Note that this algorithm in combination with \cref{lem:planarization} yields
a $3$-bend drawing of~$H$. 

For a $2$-bend drawing, around each dummy vertex of $G$, we carefully assign ports so~that each of the two edges of $H$ involved in the crossing receives two bends in total. In our approach, a particular edge incident to the dummy vertex (i.e., its last one) is drawn bendless (see \cref{fig:placement-dummy}).
We first describe for each vertex of~$G$ the ports that its incident edges use in $\Gamma$; see \cref{fig:ports-real}. The edge $(v_1,v_2)$ is assigned to an arbitrary bottom port both at~$v_1$ and~$v_2$. Consider any real vertex $v_i$ of $G$.
Then, the sole incoming blue (green) edge incident to $v_i$, if any, uses its horizontal left (right) port if its other endpoint is real, and the first bottom port in counterclockwise (clockwise) direction if its other endpoint is dummy.
If $v_i$ belongs to a chain $P_j$, the black edge to its neighbor in $P_j$ that is towards $v_1$ ($v_2$) along $C_i$, if any, uses the horizontal left (right) port at $v_i$.
Note that after assigning the left and right ports to the edges of the chain, $v_i$ has at most $\Delta-2$ outgoing edges, and $\Delta-1$ available top ports; we will leverage this observation in the biconnected case in \cref{sec:biconnected}.
The $k$-th outgoing blue (green) edge encountered in a counterclockwise (clockwise) traversal of the edges around $v_i$ in the planar embedding of $G$ uses the $k$-th top port of $v_i$ in counterclockwise (clockwise) order around $v_i$ starting from the horizontal right (left) one.
The outgoing red edge uses the $(b+1)$-th top port of $v_i$ in counterclockwise order around $v_i$ starting from the horizontal right one, where $b$ equals the number of blue outgoing edges of $v_i$.
Finally, the $k$-th incoming red edge in a clockwise traversal of the edges around $v_i$ in the planar embedding of $G$ uses the $(k+1)$-th bottom port of $v_i$ in clockwise order around $v_i$ following the horizontal right one.

\begin{figure}
    \centering
    \begin{subfigure}{0.24\textwidth}
        \centering
        \includegraphics[page=1]{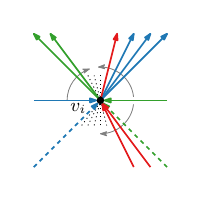}
        \subcaption{}
        \label{fig:ports-real}
    \end{subfigure}
    \hfill
    \begin{subfigure}{0.24\textwidth}
        \centering
        \includegraphics[page=2]{figures/ports.pdf}
        \subcaption{}
        \label{fig:ports-dummy-deg3}
    \end{subfigure}
    \hfill
    \begin{subfigure}{0.24\textwidth}
        \centering
        \includegraphics[page=3]{figures/ports.pdf}
        \subcaption{}
        \label{fig:ports-dummy-deg2a}
    \end{subfigure}
    \hfill
    \begin{subfigure}{0.24\textwidth}
        \centering
        \includegraphics[page=4]{figures/ports.pdf}
        \subcaption{}
        \label{fig:ports-dummy-deg2b}
    \end{subfigure}    
    \caption{Illustration of the port assignment around \textbf{(a)}~a real vertex and \textbf{(b--d)}~a dummy vertex. For a real vertex, we assign ports based on the colors and directions of the edges; for a dummy vertex, we assign ports based on the circular order of its neighbors, starting from the last one.}
    \label{fig:ports}
\end{figure}

We next proceed to assign ports around each dummy vertex. To this end, consider a dummy vertex $v_d$ with $d\neq n$; see \cref{fig:ports-dummy-deg3,fig:ports-dummy-deg2a,fig:ports-dummy-deg2b}.
Let $(v_i,v_j,v_{i'},v_{j'})$ be the neighbors of $v_d$ in counterclockwise order, such that $v_{j'}$ is the last neighbor of $v_d$.
Since $d\neq n$, the edge $(v_d,v_{j'})$ is an incoming edge at $v_{j'}$. Let $\pi$ be the bottom port that $(v_d,v_{j'})$ uses at $v_{j'}$. Then, $(v_j,v_d)$ uses $\pi$, and $(v_d,v_{j'})$ uses the port opposite of $\pi$ at $v_d$. The edge $(v_i,v_d)$ uses the left horizontal port and $(v_{i'},v_d)$ uses the right one at $v_d$.
Finally, consider the case that $v_n$ is a dummy vertex. Let $v_i,v_j,v_{i'},v_{j'}$ be the neighbors of $v_d$ in counterclockwise order such that $v_i=v_1$. In this case, there is no outgoing edge. We choose~$\pi$ as the port opposite of the port assigned to $(v_j,v_n)$ at~$v_j$ and assign the ports exactly as in the previous case. This completes the port assignment, which in the next lemma we show that is valid.

\begin{lemma}\label{lm:assignement}
No two edges share a port and the planar embedding of $G$ is preserved.   
\end{lemma}
\begin{proof}
Since the port assignment preserves the order of the edges around each vertex, the planar embedding of $G$ is preserved.

At every dummy vertex, we have assigned two edges to the horizontal ports, and two edges to opposite non-horizontal ports, so no port is used twice.

Consider a real vertex $v_i$ for $i\in \{3,\ldots,n-1\}$. It has exactly one outgoing red edge. Furthermore, it either has an incoming blue edge or 
it belongs to a chain $P_j$ and has a black edge to its neighbor in $P_j$ that is towards $v_1$ along $C_i$, which uses its left horizontal port; and it either has an incoming green edge or it belongs to a chain $P_j$ and has a black edge to its neighbor in $P_j$ that is towards $v_2$ along $C_i$, which uses its right horizontal port. Thus, there are at most $\Delta-3$ incoming red edges, and at most $\Delta-2$ outgoing edges in total. We have reserved $\Delta-3$ bottom ports for the incoming red edges and $\Delta-1$ top ports for the outgoing edges. Thus, there are enough ports for the edges to be assigned to, and the order of assignment ensures that no port is used twice.

    Vertex~$v_1$ has no incoming edge. The outgoing edge $(v_1,v_2)$ is assigned to a bottom port, so there are $\Delta-1$ top ports left for the other (at most) $\Delta-1$ outgoing edges of $v_1$. Similarly,~$v_2$ has only one incoming edge, $(v_1,v_2)$, which is assigned to a bottom port, so there are also enough top ports for its outgoing edges.

    Finally, vertex~$v_n$ has no outgoing edges. If~$v_n$ is real and has degree~$\Delta$, then we have to make sure that there are enough bottom ports for its incoming edges. Since~$v_1$ is a real vertex, the incoming blue edge $(v_1,v_n)$ uses the horizontal left port at~$v_n$. The incoming green edge uses either the horizontal right port or the first bottom port in clockwise order after the horizontal right port. Hence, there are $\Delta-2$ bottom ports remaining that the $\Delta-2$ incoming red edges are assigned to.
\end{proof}

\noindent For $0 \leq k < m$, let $\Gamma_k$ be a planar drawing of $G_k$ with the following invariant~properties:

\begin{enumerate}[label=I.\arabic*]
    \item\label{inv:cut}There is a cut through every contour edge,
except the last edge of each dummy vertex.
    \item\label{inv:ports} Each edge in $\Gamma_k$ uses its assigned port at its two endpoints.
    \item\label{inv:ymonotone} Each edge $(v_i,v_j)$ with $v_i\prec v_j$ and $j<n$ is drawn $y$-monotone from $v_i$~to~$v_j$.

    \item\label{inv:dummy} For each dummy vertex $v_d$ in $G_k$ due to the crossing of $(v_i,v_{i'})$ and $(v_j,v_{j'})$ such that $v_{j'}$ is its last neighbor, the following hold.
        \begin{enumerate}[label=\alph*)]
            \item If $v_i\in G_k$, then $(v_d,v_i)$ is drawn with at most one bend in $\Gamma_k$;
            \item If $v_{i'}\in G_k$, then $(v_d,v_{i'})$ is drawn with at most one bend in $\Gamma_k$;
            \item If $v_{j'}\in G_k$, then $(v_d,v_{j'})$ is drawn without bends in $\Gamma_k$.
        \end{enumerate}
\end{enumerate}

In the base of our recursive algorithm, $k=0$ holds and drawing $\Gamma_0$ of $G_0$ is easy to be derived. We place $v_1$ at $(0,0)$ and $v_2$ at $(2\delta+1,0)$.
Let $s_1$ and $s_2$ denote the slopes of the (bottom) ports reserved for $(v_1,v_2)$ at $v_1$ and $v_2$, respectively.
The edge $(v_1,v_2)$ is drawn with three segments: one segment of slope $s_1$ incident to $v_1$, followed by a horizontal segment at $y$-coordinate $-1$, and a segment of slope $s_2$ incident to $v_2$. Since all slopes have run between $-\delta$ and $\delta$, the resulting drawing~$\Gamma_1$ of~$G_1$ is planar and all invariants are maintained.

Next, we describe how to obtain a drawing $\Gamma_{k+1}$ of $G_{k+1}$ that satisfies \ref{inv:cut}--\ref{inv:dummy} from a corresponding drawing $\Gamma_k$ of $G_k$.
We first introduce some notation.
Suppose that $P_{k+1}=\{v_g,\ldots,v_h\}$.
Let $L$ be a horizontal line at least $h-g+1$ units above~$\Gamma_k$.
Let $w_1, \ldots, w_p$ be the neighbors of the vertices in $P_{k+1}$ on $C_k$ as encountered in a traversal of $C_k$ from $v_1$ to $v_2$, so that $w_1$ and $w_p$ are the first and last such neighbors.
If $P_{k+1}$ is a singleton, then $g=h$.
If $P_{k+1}$ is a chain, then $p=2$ holds, and the only edges between the vertices of $P_{k+1}$ and $G_i$ are $(w_1,v_g)$ and $(w_p,v_h)$.
For $1\leq i\leq p$, let $s_i$ be the slope of the port reserved for $(w_i,v_g)$
(or $(w_p,v_h)$ if $i=p$) at $w_i$ if $w_i$ is real, or at $v_g$ (resp.\ $v_h$)
otherwise, and let $s_i'$ be the slope of the port reserved for $(w_i,v_g)$
(resp.\ $(w_p,v_h)$) at $v_g$ (resp.\ $v_h$). Let $q_i$ be the intersection of $L$
with the ray at $w_i$ with slope $s_i$; see~\cref{fig:placement-deg2real-before}.
We next present a few auxiliary lemmas.

    \begin{figure}
        \centering
        \subcaptionbox{\label{fig:attachable}}{\includegraphics[page=1]{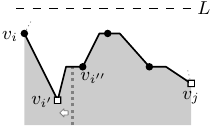}}
        \hfill
        \subcaptionbox{\label{fig:stretching}}{\includegraphics[page=1]{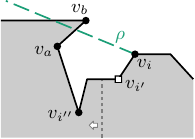}\quad \includegraphics[page=2]{figures/stretching.pdf}}
        \hfill
        \subcaptionbox{\label{fig:placement-deg2real-before}}{\includegraphics[page=16]{figures/placement.pdf}}
        \caption{Illustrations for the proofs of \textbf{(a)} \cref{lem:attachable}, \textbf{(b)} \cref{lem:rays}, and \textbf{(c)} \cref{lem:stretch}.}
        \label{fig:stretching-lemmas}
    \end{figure}

\begin{lemma}\label{lem:attachable}
    Let~$v_i$ and~$v_j$ be two vertices on the contour~$C_k$ of~$\Gamma_k$, such that each of them has a neighbor in $G \setminus G_k$. Then there is a cut through an edge between $v_i$ and $v_j$ on $C_k$.
\end{lemma}
\begin{proof}
    Let $v_{i'}$ be the neighbor of $v_i$ on the path from $v_i$ to $v_j$ along the contour $C_k$; see \cref{fig:attachable}. By \ref{inv:cut}, there is a cut through the edge $(v_i,v_{i'})$, unless it is the last edge of a dummy vertex. In the exceptional case, since each of $v_i$ and $v_j$ has a neighbor in $G \setminus G_k$, $v_{i'}$ must be the dummy vertex and $v_{i'}\neq v_g$. Then the other neighbor $v_{i''}$ of $v_{i'}$ in $C_k$ must be real, which by \ref{inv:cut} implies that there is a cut through the edge $(v_{i'}, v_{i''})$.
\end{proof}

\begin{lemma}\label{lem:rays}
    Let $v_i$ be a vertex on $C_k$ of $\Gamma_k$ having a neighbor in $G \setminus G_k$, and $s$ be the slope of a top free port of $v_i$ incident to the outer face of $\Gamma_k$.~Then a drawing $\Gamma'_k$ of~$G_k$ satisfying \ref{inv:cut}--\ref{inv:dummy} can be computed so that the ray~$\rho$ starting at $v_i$ with slope $s$ does not intersect $\Gamma'_k$.
\end{lemma}
\begin{proof}
    If~$\rho$ intersects an edge of~$G_k$, then it also intersects some edge $(v_a,v_b)$ of~$C_k$; see \cref{fig:stretching}. Assume w.l.o.g.\ that $v_a$ and $v_b$ lie on a traversal from $v_1$ to $v_i$ on $C_k$; the other case is symmetric. Let $v_{i'}$ be the neighbor of $v_i$ on the path from $v_i$ to $v_1$ along~$C_k$. By \ref{inv:cut}, there is a cut through $(v_i,v_{i'})$ unless it is the last edge of a dummy vertex. In the latter case, since $v_i$ has a neighbor in $G \setminus G_k$, $v_{i'}$ must be the dummy vertex. By \ref{inv:ymonotone}, $v_{i'}$ must lie below $v_i$. Consider now the other neighbor $v_{i''}$ of $v_{i'}$. Since $(v_{i'},v_{i''})$ is not a last edge, it contains a horizontal segment; in particular, by the port assignment, this horizontal segment is incident to $v_{i'}$. So there is a cut through this segment that lies between $v_i$ and the intersection point of $\rho$ with the edge $(v_a,v_b)$. In both cases, we have found a cut that guarantees that $\Gamma_k$ can be stretched horizontally so that the horizontal distance between $v_i$ and both $v_a$ and $v_b$ increases sufficiently to ensure that $\rho$ does not intersect $(v_a,v_b)$.
\end{proof}

\noindent In view of \cref{lem:rays}, w.l.o.g.\ we  assume  that in $\Gamma_k$ any ray emanating from a vertex on $C_k$ with the slope of a free top port does not intersect $\Gamma_k$ if the vertex has a neighbor in $G\setminus G_k$.
Recall that $q_i$ is the intersection of $L$
with the ray at $w_i$ with slope $s_i$.

\begin{lemma}\label{lem:stretch}
    Given a positive integer $d$, we can compute a drawing $\Gamma'_k$ of $G_k$ that satisfies \ref{inv:cut}--\ref{inv:dummy} such that (i)~$L$ is above $\Gamma'_k$ and (ii)~$x(q_{i+1})\ge x(q_i)+d$ for $i\in\{1,\ldots,p-1\}$.
\end{lemma}
\begin{proof}
    For each $i\in\{1,\ldots,p-1\}$, since both $w_i$ and $w_{i+1}$ have neighbors in
    $G \setminus G_k$, \cref{lem:attachable} yields a cut through an edge between $w_i$ and $w_{i+1}$ on $C_k$; see \cref{fig:placement-deg2real-before}.~Each~such cut guarantees that~$\Gamma_k$ can be stretched
    horizontally between $w_i$ and $w_{i+1}$ so as to ensure $x(q_{i+1})\ge x(q_i)+d$.
    Applying these stretches for all $i\in\{1,\ldots,p-1\}$, we obtain the
    desired spacing. Since these operations do not change any $y$-coordinates, $L$
    remains above the drawing.
\end{proof}

\noindent We next describe how to introduce $P_{k+1}$ in $\Gamma_k$ to derive $\Gamma_{k+1}$ depending on the type of~$P_{k+1}$.
We will maintain \ref{inv:ports}--\ref{inv:dummy} by construction. 
The next lemma ensures that \ref{inv:cut} is maintained for most cases, so we will only show how to maintain \ref{inv:cut} when its conditions are not met. In particular, we only create a horizontal segment below $L$ in the case that we draw an edge from a dummy vertex to a real vertex that is not its last neighbor.

\begin{lemma}\label{lem:cut-horizontal}
For every edge on $C_{k+1}$ that has a horizontal segment on or above $L$, \ref{inv:cut} holds.
\end{lemma}
\begin{proof}
Let $e$ be such an edge. W.l.o.g.\ we assume that $e \in G_{k+1}\setminus G_{k}$. 
By \cref{lem:attachable}, there is a cut through an edge between $w_1$ and
    $w_2$ on $C_k$. Symmetrically, there is also one between $w_{p-1}$ and $w_p$. Since the horizontal segment of $e$ lies on or above $L$ (i.e., above  $\Gamma_k$), one of these
    cuts can always be extended vertically to pass through the horizontal segment of
    $e$. 
\end{proof}

\begin{figure}[t]
    \centering
    \includegraphics[page=1]{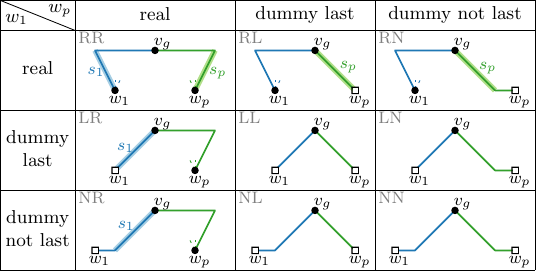}
    \caption{Introducing a degree-$2$ real singleton vertex $v_g$ in $\Gamma_k$ based on its neighbors in $C_k$.}
    \label{fig:deg-2-singleton}
\end{figure}

\medskip\noindent\textbf{$\mathbf{P_{k+1}}$ is a degree-2 singleton in $\mathbf{G_{k+1}}$ such that $\mathbf{v_g}$ is real:}
To conform with the port assignment held at the beginning of our algorithm, we seek to draw each of the edges $(w_1,v_g)$ and $(w_p,v_g)$ as follows; see \cref{fig:deg-2-singleton}.
If $w_1$ is real, then the segment of edge $(w_1,v_g)$ incident to $w_1$ will be of slope $s_1$ while the one incident to $v_g$ will be horizontal; see \cref{fig:deg-2-singleton}$\color{defblue}_{\text{RR}}$.
If~$w_1$ is dummy and $(w_1,v_g)$ is the last edge of $w_1$, then the edge $(w_1,v_g)$ will be drawn with a single segment of slope $s_1$; see \cref{fig:deg-2-singleton}$\color{defblue}_{\text{LR}}$.
Otherwise, if~$w_1$ is dummy and $(w_1,v_g)$ is not the last edge of $w_1$, then the segment of edge $(w_1,v_g)$ incident to $w_1$ will be horizontal while the one incident to $v_g$ will be of slope $s_1$; see \cref{fig:deg-2-singleton}$\color{defblue}_{\text{NR}}$.
Symmetrically, we will draw edge~$(w_p,v_g)$.

We next claim that the top port of $w_1$ (symmetrically of~$w_p$) with slope $s_1$ ($s_p$) is a free port on the outer face of $\Gamma_k$.
If $w_1$ is dummy and $(w_1,v_g)$ is not the last edge of $w_1$, all top ports of $w_1$ are free in $\Gamma_k$. Otherwise, the claim follows from the port assignment.

If $w_1$ is dummy and $(w_1,v_g)$ is not its last edge, before introducing $v_g$ we first have to stretch $\Gamma_k$ in order to place the horizontal segment of the edge $(w_1,v_g)$. To this end,
consider the neighbor $w_1'$ of $w_1$ on the path between $w_1$ and $v_2$ on $C_k$; see \cref{fig:dunno}. Since $w_1$ is dummy, $w_1'$ is real and there is a cut through $(w_1,w_1')$ by \ref{inv:cut}. We stretch $\Gamma_k$ such that everything to the left of this cut (including $w_1$) is moved to the left by one unit. Consider the ray $\rho_1$ exactly at the position it was before the stretching, i.e., one unit to the right of $w_1$ after we stretched. Since we stretched at an edge incident to $w_1$, no part of the drawing can be intersected by $\rho_1$ after the stretching. Thus, we can draw $(w_1,v_g)$ with a horizontal segment of length~1 at $w_1$ followed by a segment of slope $s_1$ on $\rho_1$. The case that $w_p$ is dummy and $(w_p,v_g)$ is not its last edge~is~symmetric.

\begin{figure}
    \centering
    \subcaptionbox{\label{fig:dunno}}{\includegraphics[page=14]{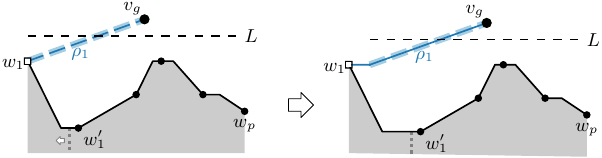}}

    \subcaptionbox{\label{fig:placement-deg2real-after}}{\includegraphics[page=2]{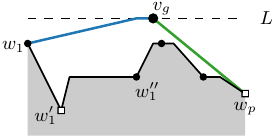}}
     \hfil
    \subcaptionbox{\label{fig:placement-deg2dummy}}{\includegraphics[page=3]{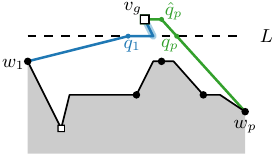}}
    \caption{Introduction of $v_g$ in $\Gamma_k$ when \textbf{(a)}~$w_1$ is dummy, and $v_g$ is real but not the last neighbor of $w_1$; \textbf{(b)} $w_1$ and $v_g$ are real; and \textbf{(c)} $w_1$ is real and $v_g$ is dummy.}
\end{figure}

If $q_p$ is to the left of $q_1$, then we apply \cref{lem:stretch} to guarantee that $q_p$ appears by at least one unit to the right of $q_1$.
To obtain $\Gamma_{k+1}$, we introduce $v_g$ at $\Gamma_k$ as follows (\cref{fig:placement-deg2real-after}).
If~$w_1$~is dummy and $w_p$ is real, we place $v_g$ at $q_1$;
if $w_1$ is real and $w_p$ is dummy, we place $v_g$ at~$q_p$;
if $w_1$ and $w_p$ are both real, we place $v_g$ along $L$ between $q_1$ and $q_p$;
if $w_1$ and $w_p$ are both dummy, we place $v_g$ at the intersection $q$ of $\rho_1$ and $\rho_p$. In the latter case, slopes $s_1$ and $s_p$ correspond to the ones of the first and the last bottom port at~$v_g$ in counterclockwise direction, respectively. 
As $q_1$ lies to the left of $q_p$, the rays $\rho_1$ and $\rho_p$ cross, and~$q$ lies above~$L$.

When $w_1$ is real, \ref{inv:cut} is preserved for $(w_1,v_g)$ by \cref{lem:cut-horizontal}, since it has a horizontal
segment on $L$. If $w_1$ is dummy and $(w_1,v_g)$ is not its last edge, consider
the (real) neighbor $w_1'$ of $w_1$ on the path from $w_1$ to $w_p$ on $C_k$. By
\ref{inv:cut}, there is a cut through $(w_1',w_1)$. By \ref{inv:ymonotone}, the horizontal segment
of this edge lies below $w_1$, so we can extend the cut through the horizontal
segment of $(w_1,v_g)$. The argument for $(w_p,v_g)$ is symmetric.

\medskip\noindent\textbf{$\mathbf{P_{k+1}}$ is a singleton such that $\mathbf{v_g}$ is dummy:}
We first discuss the case that $v_g$ has degree~2 in $G_{k+1}$. Afterwards, we prove the cases where $v_g$ has degree~3 or~4; see \cref{fig:deg-3-dummy} and \cref{fig:last-vertex}. Note that the latter case can only happen at $P_{k+1}=P_m=\{v_n\}$.

\noindent \textbf{$\mathbf{v_g}$ has degree~2 in $G_{k+1}$:} Let $(w_1,w_1')$ and $(w_p,w_p')$ be the edges of $H$ that induced $v_g$ in $G$. It follows that $w_1$, $w_p$, $w_1'$ and $w_p'$ are real. Assume w.l.o.g.\ that $(v_g,w_1')$ is the last edge of $v_g$; the case that $(v_g,w_p')$ is the last is symmetric.
By the port assignment, ports $s_1$ and $s_p$ are free and on the outer face of $\Gamma_k$ (see \cref{lm:assignement}).
We draw the edge $(w_1, v_g)$ with three segments: a segment of slope $s_1$, followed by a horizontal segment, and then a segment of slope $s_{1'}$. We also draw the edge $(w_p, v_g)$ with a segment of slope $s_p$ followed by a horizontal segment; see \cref{fig:deg-2-dummy}.
We ensure that $q_p$ appears by at least $2\delta+2$ units to the right of $q_1$ by applying \cref{lem:stretch}.
\begin{figure}
    \centering
    \subcaptionbox{\nolinenumbers \label{fig:deg-2-dummy}}{\includegraphics[page=3]{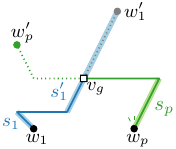}}
    \hfil
    \subcaptionbox{\nolinenumbers \label{fig:deg-3-dummy}}{\includegraphics[page=5]{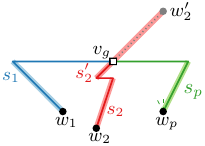}}
    \caption{Illustration of the case that $v_g$ is a dummy singleton of \textbf{(a)} degree-2 and \textbf{(b)} degree-3.}
    \label{fig:placement-dummy}
\end{figure}
We place $v_g$ one unit above $L$ at $x$-coordinate $(x(q_1)+x(q_p))/2$; see \cref{fig:placement-deg2dummy}. This ensures that we can place the horizontal segment of $(w_1,v_g)$ on $L$ and connect it to $v_g$ with a segment of slope $s_{1'}$. This last segment has horizontal length at most~$\delta$, so it connects to the horizontal segment on a point strictly between $q_1$ and $q_p$ on~$L$. We draw $(w_p,v_g)$ with a segment of slope $s_p$ starting at $w_p$ and ending at a point $\hat q_p$ with $y(\hat q_p)=y(v_g)$. Since the horizontal distance between $q_p$ and $\hat q_p$ is at most~$\delta$, $\hat q_p$ lies strictly to the right of $v_g$, allowing us to complete the edge with a horizontal segment. Thus, the edges $(w_1,v_g)$ and $(w_p,v_g)$ are drawn as desired and do not cross each other.
Observe that the first segments of these two edges are crossing-free as described, and the other segments lie above $\Gamma_k$, so $\Gamma_{k+1}$ is planar. 

\noindent \textbf{$\mathbf{v_g}$ has degree~3 in $\mathbf{G_{k+1}}$:}
Let $(w_1,w_p)$ and $(w_2,w_2')$ be the edges of $H$ that induce $v_g$, where $w_2$ lies on $C_k$. Let $s_2$ and $s_2'$ be the slopes reserved for $(w_2,v_g)$ and $(w_2',v_g)$ at $w_2$ and $w_2'$, respectively. We draw $(w_1,v_g)$ and $(w_p,v_g)$ with two segments (a segment of slope $s_1$ or $s_p$, then horizontal), and $(w_2,v_g)$ with three segments (slope $s_2$, horizontal, slope $s_2'$); see \cref{fig:deg-3-dummy}. By construction, these ports are free and lie on the outer face of $\Gamma_k$.
Let $q_2$ be the intersection of $L$ and the ray from $w_2$ with slope $s_2$; see the left part of \cref{fig:strecth-deg3}. By \cref{lem:stretch}, we ensure that $x(q_2)-x(q_1)$ and $x(q_p)-x(q_2)$ are at least $2\delta+1$; see the right part of \cref{fig:strecth-deg3}.
We place $v_g$ one unit above $L$ at $x$-coordinate $x(q_2)$ (or $x(q_2)+1$ if $s_2'$ is vertical). Hence, we can route the horizontal segment of $(w_2, v_g)$ along $L$ and connect it to $v_g$ with slope $s_2'$. The final segment of this edge has horizontal length at most~$\delta$, ensuring that it meets the horizontal segment at a point strictly between $q_1$ and $q_p$ on~$L$. Next, we draw the edges $(w_p,v_g)$ and $(w_1,v_g)$ as in the degree-2 dummy case.
Thus, the edges $(w_1,v_g)$, $(w_2,v_g)$, and $(w_p,v_g)$ are drawn as desired and do not cross each other.
As before, the initial segments of the three new edges are crossing-free, and the remaining segments (including the horizontal) lie above $\Gamma_k$. Hence, $\Gamma_{k+1}$ is planar.

\noindent \textbf{$\mathbf{v_g}$ has degree~4 in $\mathbf{G_{k+1}}$:}
By the properties of a canonical order, this can only happen at $P_{k+1}=P_m=\{v_n\}$, so we do not have to maintain \ref{inv:ymonotone}, we have $w_1=v_1$.
Recall that we have assigned the horizontal slope to $(w_1,v_n)$ and $(w_3,v_n)$ at $v_n$, and the slope $s_2$ to $(w_p,v_n)$ and $(w_2,v_n)$ at $v_n$; see \cref{fig:last-vertex}.
We place $v_n$ at the intersection of $L$ with the ray starting in $w_2$ with slope $s_2$. We draw the edge $(w_2,v_n)$ with one segment of slope $s_2$. We draw the edge $(w_p,v_n)$ with 3 segments: the first starts in $w_p$ with slope $s_p$; the second is horizontal one unit above $L$; and the third ends in $v_n$ with slope $s_2$. The edges $(w_1,v_n)$ and $(w_3,v_n)$ are both drawn with two segments: one with slope $s_1$ and one with $s_3$ at $w_1$ and $w_3$, respectively, followed by a horizontal segment at $v_n$.
All segments, except the first segment of the new edges, lie above $\Gamma_{m-1}$, and they do not cross each other, so $\Gamma=\Gamma_{m}$ is planar.

\begin{figure}[t]
    \centering
    \subcaptionbox{\label{fig:chain-dummy1}}{\includegraphics[page=8]{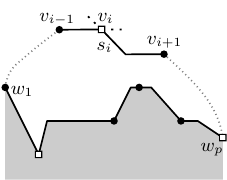}}
    \hfil
    \subcaptionbox{\label{fig:chain-dummy2}}{\includegraphics[page=7]{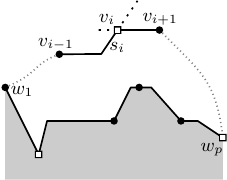}}
    \hfil
    \subcaptionbox{\label{fig:chain-real}}{\includegraphics[page=6]{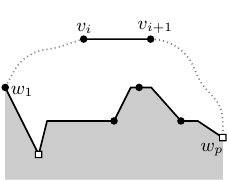}}
    \caption{Introducing a chain in $\Gamma_k$: \textbf{(a)}~$v_i$ is dummy, $(v_i, v_{i+1})$ is a critical edge,
    \textbf{(b)}~$v_i$ is dummy, $(v_i, v_{i+1})$ is a non-critical edge, and
    \textbf{(c)}~$v_i$ and $v_{i+1}$ are real.}
    \label{fig:placement-chain}
\end{figure}

\medskip\noindent\textbf{$\mathbf{P_{k+1}}$ is a chain in $\mathbf{G_{k+1}}$:}
Suppose that $P_{k+1}=\{v_g,\ldots,v_h\}$ is a chain in $G_{k+1}$. 
Since $G$ is $1$-planar, it follows that no two consecutive vertices of $P_{k+1}$ are dummy vertices. In other words, this implies that every dummy vertex of $P_{k+1}$ is preceded and followed by real vertices in $P_{k+1} \cup \{w_1,w_p\}$.

In view of \cref{lem:stretch}, we may assume without loss of generality that $q_p$ appears by at least $(3\delta+1) \cdot (h-g)+2$ units to the right of $q_1$. 
To obtain drawing $\Gamma_{k+1}$ of $G_{k+1}$ satisfying  \ref{inv:cut}--\ref{inv:dummy}, we sequentially introduce in $\Gamma_k$ vertices $v_g,\ldots,v_h$ of $P_{k+1}$ in this order.
Initially, we introduce $v_g$ in $\Gamma_k$ as it were a degree-2 singleton (real or dummy, as appropriate), placing it so that its $x$-coordinate coincides with that of $q_1$ if $w_1$ is dummy and the edge $(w_1,v_g)$ is the last edge of $w_1$ or with $q_1+1$, otherwise.
Assume that for some $i$ in $\{g,\ldots,h-1\}$ we have processed the vertices $v_g,\ldots,v_i$ such that the horizontal distance between any two consecutive vertices is exactly $2\delta+1$, while the corresponding vertical distance is at most~$1$. Let $v_{i+1}$ be the next vertex of $P_{k+1}$ to process. 

We first consider the case that $i+1<h$.
We set the $x$-coordinate of $v_{i+1}$ to $x(v_{i})+2\delta+1$. The $y$-coordinate of $v_{i+1}$ depends on the edge $(v_{i},v_{i+1})$ and also on the types of vertices $v_i$ and $v_{i+1}$.

Assume that $v_i$ is a dummy vertex. If $(v_i,v_{i+1})$ is a critical edge (see \cref{fig:chain-dummy1}), then we set the $y$-coordinate of $v_{i+1}$ to $y(v_i)-1$. We further draw the edge $(v_i,v_{i+1})$ with a horizontal segment of length at least $\delta+1$ incident to $v_{i+1}$ followed by a segment of slope 
$s_i$ incident to $v_i$, where $s_i$ is the slope of the edge $(v_i,v_{i+1})$ reserved at $v_i$. Otherwise (see \cref{fig:chain-dummy2}), the edge $(v_i,v_{i+1})$ is not a critical edge. In this case, we set the $y$-coordinate of $v_{i+1}$ to $y(v_i)$ and we draw the edge $(v_i,v_{i+1})$ as a horizontal edge. 

Assume now that $v_i$ is a real vertex. In this case, $v_{i+1}$ can either be a real vertex or a dummy vertex. In the first scenario (see \cref{fig:chain-real}), we set the $y$-coordinate of $v_{i+1}$ to $y(v_i)$ and we draw the edge $(v_i,v_{i+1})$ as a horizontal segment. Otherwise, $v_{i+1}$ is a dummy vertex. 
In this case, if $(v_i,v_{i+1})$ is a critical edge (see the edge $(v_{i-1}, v_i)$ in \cref{fig:chain-dummy2}), then we set the $y$-coordinate of $v_{i+1}$ to $y(v_i)+1$. We further draw the edge $(v_i,v_{i+1})$ with a horizontal segment of length at least $\delta+1$ incident to $v_i$ followed by a segment of slope $s_{i+1}$ incident to $v_{i+1}$, where $s_{i+1}$ is the slope of the edge $(v_i,v_{i+1})$ reserved at $v_{i+1}$. On the other hand (see the edge $(v_{i-1}, v_i)$ in  \cref{fig:chain-dummy1}), if the edge $(v_i,v_{i+1})$ is not critical, we set the $y$-coordinate of $v_{i+1}$ to $y(v_i)$ and we draw the edge $(v_i,v_{i+1})$ as a horizontal edge.

Consider now the vertex $v_h$. The $y$-coordinate of $v_h$ is computed in the same way as for all other vertices on the chain. We then first draw the edge $(v_h,w_p)$ exactly as in the real or dummy singleton case, depending on whether $v_h$ is real or dummy. By definition of $q_p$, if $v_h$ is placed on $L$, then its $x$-coordinate is either $x(q_p)$ or $x(q_p)-1$. Since $|y(q_p)-y(L)|\le h-g$, we have that 
\begin{align*}
x(v_h)&\ge x(q_p)-s_p(h-g)-1\\
&\ge x(q_p)-\delta(h-g)-1\\
&\ge x(q_1)+(3\delta+1)(h-g)+2-\delta(h-g)-1\\
&= x(q_1)+(2\delta+1)(h-g)+1\\
&\ge x(v_g)+(2\delta+1)(h-g)\\
&= x(v_{h-1})-(2\delta+1)(h-g-1)+(2\delta+1)(h-g)\\
&= x(v_{h-1})+2\delta+1.
\end{align*}
Thus, we can draw the edge $(v_{h-1},v_h)$ exactly as the other edges on the chain, extending its horizontal segment as needed.

The construction above gives us the desired drawings of the edges and maintains \ref{inv:ports},~\ref{inv:ymonotone} and~\ref{inv:dummy}. 
Since the edges $(w_1,v_g)$ and $(w_p,v_h)$ are drawn as in the degree-2 singleton cases and all edges of the chain lie above $\Gamma_k$, $\Gamma_{k+1}$ is planar. 
By \cref{lem:attachable}, there is a cut between $w_1$ and $w_p$ on $C_k$, and this cut can be extended to every edge on the chain, thereby maintaining~\ref{inv:cut}.

\medskip

\begin{figure}[t]
    \centering
    \subcaptionbox{\nolinenumbers \label{fig:strecth-deg3}}{
    \includegraphics[page=15]{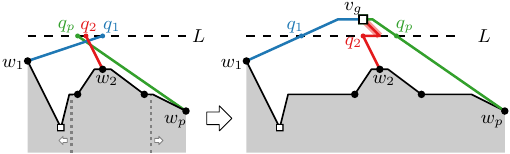}}
    \hfill
    \subcaptionbox{\nolinenumbers \label{fig:last-vertex}}{\includegraphics[page=13]{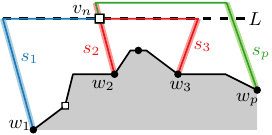}}
    \caption{\textbf{(a)} Introducing a degree-3 dummy singleton in $\Gamma_k$, and \textbf{(b)} placement of dummy $v_n$.}
\end{figure}

\medskip\noindent\textbf{$\mathbf{P_{k+1}}$ is a singleton of degree $\mathbf{p\ge 3}$ in $\mathbf{G_{k+1}}$ such that $\mathbf{v_g}$ is real:}
\label{ssec:real-highdeg}
To cope with this case, we consider three~subcases.

\smallskip\noindent \textit{Case 1: $w_1$ and $w_p$ are dummy.}
Recall that $(v_1,v_n)$ is an edge and $v_1$ is real, so this case does not occur for $g=n$.
Assume first that $v_g$ is the last neighbor of $w_1$ and $w_p$. By \cref{lem:stretch}, we ensure that $q_1,\ldots,q_p$ appear on $L$ in this left-to-right order with pairwise horizontal distance at least~2.
We fix the first segment of each edge $(w_i,v_g)$, namely the one from $w_i$ to $q_i$.
For each $i$, let $\rho_i'$ be the ray from $q_i$ with slope $s_i'$, if $w_i$ is dummy; otherwise, $\rho_i'$ is the ray from $(x(q_i)+1,y(q_i))$ with slope $s_i'$; see \cref{fig:beautiful-before}. 
By the way we have assigned the ports at $v_g$, every pair of rays $\rho'_i$ and $\rho'_j$ cross.
Further, this choice of $\rho_i'$ guarantees that each edge $(w_i,v_g)$ will be drawn with a single segment when $w_i$ is dummy, and otherwise with three segments, one of which (the middle one) is horizontal along $L$.
For $1 \leq i \leq p-1$, let $y_i$ be the $y$-coordinate of the intersection of $\rho_i'$ and $\rho_{i+1}'$.
Stretching between $w_i$ and $w_{i+1}$ translates $q_1,\ldots,q_i$, while leaving $q_{i+1},\ldots,q_p$ fixed.
Hence, such a stretch changes only $y_i$; in particular, it increases $y_i$ continuously. We next apply two sweeps (in this order).

\begin{figure}[t]
    \centering
    \subcaptionbox{\label{fig:beautiful-before}}{\includegraphics[page=9]{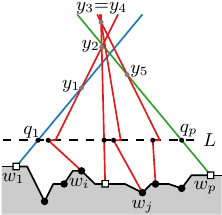}}
    \hfil
    \subcaptionbox{\nolinenumbers \label{fig:beautiful-firstsweep}}{\includegraphics[page=11]{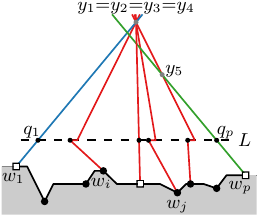}}
    \hfil
    \subcaptionbox{\nolinenumbers \label{fig:beautiful-secondsweep}}{\includegraphics[page=12]{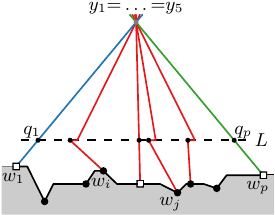}}
    \caption{Introducing a singleton of degree at least $3$ in $\Gamma_k$: (a)~initially, (b)~after the right-to-left sweep, and (c)~after the left-to-right sweep.}
    \label{fig:beautiful}
\end{figure}

\begin{itemize}
\item Right-to-left sweep:
For $i=p-2,p-3,\ldots,1$,
if $y_i< y_{i+1}$, apply \cref{lem:stretch} to $w_i$ and $w_{i+1}$ until $y_i=y_{i+1}$.
After the sweep, 
$y_1\ge y_2\ge \ldots \ge y_{p-1}$ holds; see \cref{fig:beautiful-firstsweep}.
\item Left-to-right sweep:
For $i=1,2,\ldots,p-2$,
if $y_{i+1}< y_i$, apply \cref{lem:stretch} to $w_{i+1}$ and $w_{i+2}$ until $y_{i+1}=y_i$.
After the sweep, $y_1=y_2=\ldots=y_{p-1}$ holds; see \cref{fig:beautiful-secondsweep}.
\end{itemize}
All rays $\rho_1',\ldots,\rho_p'$ meet at a single point.
Indeed, let $x_i$ be the intersection of $\rho_i'$ and $\rho_{i+1}'$.
Since all $y_i$ are equal, all $x_i$ have the same $y$-coordinate.
Also, $x_i$ and $x_{i+1}$ both lie on ray $\rho_{i+1}'$.
As a ray intersects any horizontal line in at most one point, we obtain $x_i=x_{i+1}$ for each $i$.
So $x_1=\ldots=x_{p-1}$, i.e., rays $\rho_1',\ldots,\rho_p'$ are concurrent and
$v_g$ is placed there.

If $v_g$ is not the last neighbor of $w_1$, then we apply \cref{lem:stretch} once more between $w_1$ and $w_2$ to move $w_1$ by one unit such that we can add the horizontal segment of $(w_1,v_g)$; see \cref{fig:dunno}. We handle the case that $v_g$ is not the last neighbor of $w_p$ symmetrically.

\smallskip\noindent\textit{Case 2: $w_1$ or $w_p$ are real.}
We apply the above construction to the remaining neighbors of $v_g$.
Then, we connect $v_g$ to each omitted real neighbor in $C_k$ with two segments; a first of slope $s_1$ or $s_p$, followed by a horizontal incident to $v_g$.
By a suitable stretch near the corresponding side of the contour, each horizontal segment can be placed above the current drawing, avoiding crossings. The invariant~\ref{inv:cut} is preserved for $(w_1,v_g)$ and $(w_p,v_g)$ by \cref{lem:cut-horizontal}.

\medskip

Recall that we assumed $(v_1,v_2)$ to be an edge between real vertices. 
However, there exist triconnected 1-plane graphs in which every edge is crossed. If~$H$ is such a graph, then there is no edge between two real vertices in its planarization~$G$. Thus, no such edge can be selected as $(v_1,v_2)$ in the canonical order. In this case, we proceed as follows.

\begin{figure}[b]
     \centering
    \subcaptionbox{\label{fig:face}}{\includegraphics[page=1]{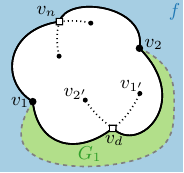}}
    \hfil
    \subcaptionbox{\label{fig:gamma1}}{\nolinenumbers \includegraphics[page=6]{figures/drawing-styles.pdf}}
    \caption{The case that all edges of $H$ are crossed. \textbf{(a)} Selecting face~$f$ and adding the dummy edge $(v_1,v_2)$; \textbf{(b)} Drawing of $\Gamma_1$.}
    \label{fig:all-crossed}
\end{figure} 

We choose an arbitrary face~$f$ as the outer face. By 1-planarity, there is no edge between two dummy vertices, and by assumption, there is no edge between two real vertices. Hence, the boundary of~$f$ contains at least four vertices; two real and two dummy. We choose $v_n,v_1,v_d,v_2$ as four consecutive vertices on the boundary of~$f$ such that $v_n$ and $v_d$ are the dummy vertices, and $v_1$ and $v_2$ are the real ones; see \cref{fig:face}.

We add an extra edge $(v_1,v_2)$ to~$G$. Observe that this may increase the degree of~$v_1$ and~$v_2$ to~$\Delta+1$, requiring a slight adaptation of our algorithm. In the canonical order, we have $P_1=\{v_d\}$, so $G_1$ consists of the triangle $(v_1,v_2,v_d)$. We assign the horizontal right port to $(v_1,v_d)$ at $v_1$ and the horizontal left port to $(v_2,v_d)$ at $v_2$. All other ports are assigned as described above.

We draw $G_1$ similar to the degree-2 dummy singleton case. Without loss of generality,  assume that $v_{1'}$ is the last neighbor of $v_d$; the other case is symmetric. 
We place $v_1$ at $(0,0)$, $v_d$ at $(2\delta,1)$ and $v_2$ at $(3\delta+1,1)$; see \cref{fig:gamma1}. Let $s_d$ be the slope of the port reserved for $(v_1,v_d)$ at $v_d$, and let $s_1$ and $s_2$ denote the slopes of the (bottom) ports reserved for $(v_1,v_2)$ at $v_1$ and $v_2$, respectively.
Edge $(v_1,v_d)$ is drawn with two segments: one horizontal at~$v_1$, and one with slope~$s_d$ at~$v_d$. Edge $(v_1,v_2)$ is drawn as a single horizontal segment. Finally, edge $(v_1,v_2)$ is drawn with three segments: one segment of slope $s_1$ incident to $v_1$, followed by a horizontal segment at y-coordinate $-1$, and a segment of slope $s_2$ incident to $v_2$. Since all slopes have rise between $-\delta$ and $\delta$, the resulting drawing~$\Gamma_1$ of~$G_1$ is planar and satisfies all invariants. See \cref{fig:drawing} for an example drawing created by our algorithm.

\begin{figure}[t]
    \centering
    \includegraphics[page=6]{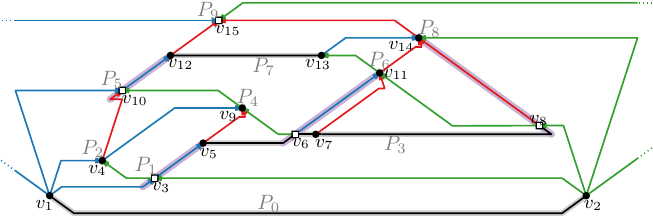}
    \caption{Illustration of a drawing created by the algorithm in \cref{thm:triconnectedcase}.}
    \label{fig:drawing}
\end{figure}

\begin{theorem}\label{thm:triconnectedcase}
For any $\Delta \geq 4$, every set $S$ of $\Delta$ slopes is universal for $2$-bend $1$-planar drawings of triconnected $1$-planar graphs with maximum degree $\Delta$. Also, for any such graph on $n$ vertices, a $2$-bend $1$-planar drawing on $S$ can be computed in $O(n)$ time if a $1$-planar embedding is specified as part of the input.
\end{theorem}
\begin{proof}
    The algorithm described above creates a $2$-bend planar drawing $\Gamma$ of a planarization~$G$ of a 1-planar drawing of the input 1-planar graph $H$ on $S$ that satisfies \ref{inv:ports} and \ref{inv:dummy}. 
    After removing the dummy vertices from $\Gamma$ and merging the incident opposite edges, we obtain a $1$-planar drawing $\Gamma'$ of~$H$. All uncrossed edges have at most two bends in~$\Gamma'$. Consider a crossed edge $(v_i,v_i')$. Let $v_d$ be the dummy vertex in $G$ that corresponds to the crossing of this edge. By \ref{inv:dummy}, $(v_i,v_d)$ and $(v_{i'},v_d)$ are drawn with at most two bends in total in $\Gamma$, and by \ref{inv:ports} and the port assignment, they use opposite ports at $v_d$. Thus, $(v_i,v_{i'})$ is drawn with at most two bends in~$\Gamma'$.
    Since the planarization of a $1$-planar embedding and the canonical order of a triconnected planar graph can be computed in $O(n)$ time, our algorithm runs in $O(n)$~time~as~well.
\end{proof}

\section{The biconnected case}
\label{sec:biconnected}

In this section, we focus on the biconnected case. Let $H$ be a biconnected $1$-plane graph and $G$ a planarization of it satisfying Properties~\ref{prp:tricon} and \ref{prp:poles} of \cref{lem:planarization}.

Let~$\mathcal{T}$ be the SPQR-tree~\cite{BattistaETT99} of~$G$. 
Assume that there exists a separation pair~$\langle s,t\rangle$ where $s$ and $t$ are real. We root$~\mathcal{T}$ in a node $\hat\mu$ with poles $s$ and $t$.
We will show how to handle the case that no such separation pair exists at the end of the section.

Intuitively, our algorithm works as follows.
We draw the nodes of the SPQR-tree bottom-up as stretchable rectangular ``chips'' with pins on their boundary that can be connected to the poles via edges with at most one extra bend; see \cref{fig:poles-addition}. At S- and P-nodes, we carefully align the chips of the child nodes, and at R-nodes we reuse the algorithm of \cref{sec:triconnected}.

Since the poles of a $P$-node are either joined by an edge or separated into at least three components, Property \ref{prp:poles} of \cref{lem:planarization} implies that no pole of a $P$-node is a dummy vertex.

Let~$\mu$ be a node of~$\mathcal{T}$ with poles~$s_\mu$ and~$t_\mu$.
If neither~$s_\mu$ nor~$t_\mu$ is a dummy vertex, we denote by~$G_\mu$ the pertinent graph of~$\mu$ (which is the subgraph of $G$ induced by the subtree of $\mathcal{T}$ rooted at $\mu$).
If~$s_\mu$ ($t_\mu$) is a dummy vertex with three neighbors in the pertinent graph of~$\mu$, then the parent~$\pi_\mu$ of~$\mu$ is an~$S$-node with an additional~$Q$-node child~$\nu$ ($\xi$) with poles~$s_\nu$ and~$t_\nu$, where~$t_\nu=s_\mu$ ($s_\xi$ and~$t_\xi$, where~$s_\xi=t_\mu$). 
We call each of these special~$Q$-nodes \emph{satellite} \emph{associated} to~$\mu$; see \cref{fig:augmented-example}. Since we seek to process~$\mu$ together with its satellites,  we denote by~$G_\mu$ the pertinent graph of~$\mu$, which if~$s_\mu$ ($t_\mu$) is a dummy vertex with three neighbors in~$G_\mu$, then we augment it with the edge~$(s_\mu, s_\nu)$ (with the edge~$(t_\mu, t_\xi)$, respectively); with slight abuse of notation, we further assume that~$s_\nu$ ($t_\xi$) is a pole of~$\mu$. We refer to such a node~$\mu$ and to such a pertinent graph~$G_\mu$ as \emph{augmented}. 
Note that, by definition and by \cref{lem:planarization}, an augmented node can be neither $Q$ nor $P$.

\begin{figure}[t]
    \centering
    \subcaptionbox{\label{fig:augmented-example}}{\includegraphics[page=1]{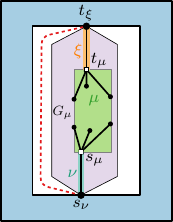}}
    \hfil
    \subcaptionbox{\label{fig:reduced-single}}{\includegraphics[page=3]{figures/reduced.pdf}}
    \hfil
    \subcaptionbox{\label{fig:reduced-degenerate}}{\includegraphics[page=2]{figures/reduced.pdf}}
    \caption{(a)~An augmented node~$\mu$ (green) with satellite $Q$-nodes~$\nu$ and~$\xi$ and augmented pertinent graph~$G_\mu$ (purple); (b)~a node~$\mu$ inside its parent $R$-node~$\pi_\mu$, reduced at~$s_\mu$, with reduced pertinent graph~$G_\mu$; (c)~a node reduced at both poles with~$u_\mu = v_\mu$. Red dashed: the virtual edge.}
    \label{fig:reduced}
\end{figure}

By Property \ref{prp:poles} of \cref{lem:planarization}, a dummy vertex that is a pole of a node of~$\mathcal{T}$ has either exactly three or exactly one neighbor in the pertinent graph of this node.
The former case is handled by the augmentation above; we now describe the symmetric treatment of the latter case.
So, let~$\mu$ be a node of~$\mathcal{T}$ whose pole~$s_\mu$ is a dummy vertex with exactly one neighbor~$u$ in the pertinent graph of~$\mu$.
Then~$s_\mu$ has degree~$2$ in the skeleton of~$\mu$, so~$\mu$ is an~$S$-node, and the first edge on the path from~$s_\mu$ to~$t_\mu$ in its skeleton is the real edge~$(s_\mu,u_\mu)$ (see \cref{fig:reduced-single}), which corresponds to a~$Q$-node child of~$\mu$.
We also call this~$Q$-node a \emph{satellite}, but associated to the parent~$\pi_\mu$ of~$\mu$, which will draw the edge~$(s_\mu,,u_\mu)$ itself.
Note that~$\pi_\mu$ is an~$R$-node:
no two~$S$-nodes are adjacent in~$\mathcal{T}$, a~$P$-node and its children share their poles, so a~$P$-node parent would have the dummy pole~$s_\mu$, contradicting Property \ref{prp:poles} of \cref{lem:planarization}, and the only~$Q$-node with a child is the root, whose poles are real.
We remove~$s_\mu$ and the edge~$(s_\mu,u_\mu)$ from the pertinent graph of~$\mu$ and assume, again with slight abuse of notation, that~$u_\mu$ is a pole of~$\mu$; if~$t_\mu$ is also a dummy vertex with exactly one neighbor~$v_\mu$ in the pertinent graph of~$\mu$, we proceed symmetrically at~$t_\mu$; see \cref{fig:reduced-degenerate} for the degenerate case~$u_\mu = v_\mu$.
We refer to such a node~$\mu$ and to such a pertinent graph~$G_\mu$ as \emph{reduced}.
Observe that the poles of a reduced node are real, and that~$\langle u_\mu,t_\mu\rangle$ is again a separation pair, unless the reduced pertinent graph consists of a single edge or, if~$\mu$ has been reduced at both poles with~$u_\mu=v_\mu$, of the single vertex~$u_\mu$; these degenerate cases are drawn directly by~$\pi_\mu$, as described in the~$R$-node case below.
 
Let~$\mu$ be a node of~$\mathcal{T}$ with poles~$s_\mu$ and~$t_\mu$ in a bottom-up traversal of $\mathcal{T}$. 
Let~$\overline{G}_\mu$ be the graph obtained from~$G_\mu$ by removing the edge~$(s_\mu,t_\mu)$, subdividing every edge that is incident to~$s_\mu$ and~$t_\mu$ once, and then removing the poles~$s_\mu$ and~$t_\mu$. 
We seek to construct a drawing of~$\overline{G}_\mu$ on~$S$, which is enclosed in an axis-parallel rectangle, such that all subdivision vertices are incident to horizontal segments and the ones previously incident to~$s_\mu$ ($t_\mu$) lie on the left (right) boundary of this rectangle. We call the rectangle containing the drawing \emph{the chip}~$C_\mu$ of~$\mu$ and the subdivision vertices on the left (right) side \emph{left} (\emph{right}) \emph{pins}. 
Observe that such a drawing is \emph{stretchable} in the sense that all edges incident to its pins can be elongated. Also, note that in contrast to the approach in~\cite{AngeliniBLM19}, we do not require the bottommost pins on the left and right side to be on the bottom left and bottom right corners of~$C_\mu$, since we can use one more slope in our approach. 
Furthermore, we will maintain as an invariant property that whenever a pin is incident to a vertex in~$\overline{G}_\mu$, then the edge connecting this pin to the pole can be drawn with one bend in a~$2$-bend drawing of the subgraph of~$H$ corresponding to~$\overline{G}_\mu$ that contains the stretchable drawing of~$\overline{G}_\mu$ as a subdrawing. In other words, one may spend one more bend to connect each pin to the pole to complete the stretchable drawing to a drawing of~$\overline{G}_\mu$. This property allows us to use the following auxiliary lemma (see \cref{fig:poles-addition}). 

\begin{lemma}[Angelini, Bekos, Liotta, Montechianni~\cite{AngeliniBLM19}]\label{lem:poles-addition}
Let $u \in \{s_\mu, t_\mu\}$ be a pole of a node $\mu$ of $\mathcal{T}$ and let $u_1, \ldots, u_q$ be the neighbors of $u$ in $\overline{G}_\mu$. Suppose that there exists a set of $\rho$ consecutive free rays of $u$
and a stretchable drawing of $\overline{G}_\mu$ such that the elongation of the edge incident to each
pin of $u$ intersects all these rays. Then, each of $(u, u_1), \ldots, (u, u_\rho)$ can be drawn with two straight-line segments whose slopes are in $S$, without introducing crossings.
\end{lemma}

\begin{figure}
	\centering
	\includegraphics[page=6]{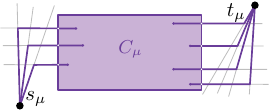} 
	\hfil  
	\includegraphics[page=7]{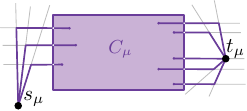} 
	\caption{Illustration of \cref{lem:poles-addition}.}
	\label{fig:poles-addition}
\end{figure}

If $\nu_1, \ldots, \nu_h$ are the children of $\mu$, we construct a stretchable drawing of $\overline{G}_\mu$  from the stretchable drawings of $\overline{G}_{\nu_1}, \ldots, \overline{G}_{\nu_h}$ that have
been recursively constructed such that for each $i = 1, \ldots, h$, the embedding of $G$ restricted to $\overline{G}_{\nu_i}$ is preserved.
We proceed by distinguishing four~cases.

\medskip\noindent\textbf{$\mu$ is a $Q$-node:}
If~$\mu$ is not the root of~$\mathcal{T}$, then~$\overline{G}_\mu$ is an empty graph; the edge~$(s_\mu, t_\mu)$ corresponding to~$\mu$ will be drawn 
by~$\mu$'s parent~$\xi$, if~$\xi$ is not a~$P$-node, or by the parent of~$\xi$ otherwise. 
If~$\mu$ is a satellite~$Q$-node, then the edge~$(s_\mu, t_\mu)$ will be drawn when its associated node of~$\mathcal{T}$ will be visited, that is, the augmented node or, for a satellite coming from a reduced node, the parent~$R$-node.
Finally, if~$\mu$ is the root, then $\mu$ has a single child~$\nu$. As~$\overline{G}_\mu = \overline{G}_{\nu}$, the stretchable drawing of~$\overline{G}_{\nu}$ serves as the stretchable drawing for~$\overline{G}_\mu$.
By \cref{lem:poles-addition},~$s_\mu$, $t_\mu$ and their incident edge, can be added to this drawing.

\begin{figure}
    \centering
    \subcaptionbox{\nolinenumbers \label{fig:p_node_2}}{\includegraphics[page=2]{figures/biconnected.pdf}}   
     \hfil  
    \subcaptionbox{\nolinenumbers \label{fig:p_node_1}}{\includegraphics[page=1]{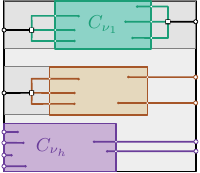}}
    \caption{Placement of the chips inside chip $C_\mu$ when $\mu$ is \textbf{(a)}~an augmented $S$-node or \textbf{(b)}~a $P$-node.}
    \label{fig:p_node}
\end{figure}

\medskip\noindent\textbf{$\mu$ is a $P$-node:}
Recall that $\mu$ has not been augmented and proceed as in \cite{AngeliniBLM19}.
We construct a new chip $C_\mu$ with a height exceeding the combined heights of $C_{\nu_1}, \ldots, C_{\nu_h}$ and a width greater than the maximum width among them; see \cref{fig:p_node_1}.
We then arrange the stretchable drawings of $\overline{G}_{\nu_1}, \ldots, \overline{G}_{\nu_h}$ such that the associated chips $C_{\nu_1}, \ldots, C_{\nu_h}$ are within $C_\mu$ and do not intersect one another, each chip has its left boundary aligned with the left boundary of $C_\mu$, and the lower boundary of $C_{\nu_h}$ coincides with the lower boundary of $C_\mu$. Subsequently, the segments incident to the pins located on the right boundaries of $C_{\nu_1}, \ldots, C_{\nu_h}$ are prolonged horizontally until they reach the right boundary of $C_\mu$.
To preserve the embedding of $G$, we may need to permute $C_{\nu_1}, \ldots, C_{\nu_h}$ so that the order in which the pins appear along the left and the right side of $C_\mu$ conforms with the embedding of $G$ restricted to $\overline{G}_\mu$. 
The resulting drawing is stretchable as each drawing associated with $C_{\nu_1}, \ldots, C_{\nu_h}$ is stretchable.

\medskip\noindent\textbf{$\mu$ is an $S$-node:} 
We first describe the case where $\mu$ has been augmented. Let $\mu$ be associated with two satellite $Q$-nodes (see \cref{fig:p_node_2}); the case of a single satellite $Q$-node is easier. Let $\sigma_\mu$ and $\tau_\mu$ be the (dummy) poles of $\mu$. We first construct a stretchable drawing as if $\mu$ were a standard node, in which the pins incident to the neighbors $\{u_1, u_2, u_3\}$ of $\sigma_\mu$ and to the neighbors $\{v_1, v_2, v_3\}$ of $\tau_\mu$ appear on the boundary of the chip according to the planar embedding. To obtain $C_\mu$, we place $\sigma_\mu$ and $\tau_\mu$ sufficiently to the left and to the right of the pins of $u_2$ and $v_2$, respectively, and we connect them to $u_2$ and $v_2$ by elongating these pins. The edges $(u_1, \sigma_\mu)$ and $(u_3, \sigma_\mu)$, as well as, $(v_1, \tau_\mu)$ and $(v_3, \tau_\mu)$ are then completed with one bend each using opposite ports at $\sigma_\mu$ and $\tau_\mu$. Finally, we extend two horizontal segments from $\sigma_\mu$ and $\tau_\mu$ to the chip boundaries to serve as the new pins for $s_\mu$ and $t_\mu$.

We now consider the case where $\mu$ has not been augmented.
If a pole of~$\mu$ is a dummy vertex, as $\mu$ has not been augmented, by Property \ref{prp:poles} of \cref{lem:planarization} it has exactly one neighbor in the pertinent graph of~$\mu$, so~$\mu$ has been reduced at this pole, the corresponding satellite~$Q$-node is drawn by the parent of~$\mu$, and the pole of the reduced node is real.
Hence, in the following,~$s_\mu$ and~$t_\mu$ denote the (real) poles of the possibly reduced node~$\mu$.
Let~$P$ be the path of virtual edges between~$s_\mu$ and~$t_\mu$ obtained by removing from the skeleton of~$\mu$ the virtual edge between its original poles as well as, if~$\mu$ has been reduced, the satellite~$Q$-nodes at its ends. Contracting every virtual edge corresponding to a satellite~$Q$-node associated to a child of~$\mu$, yields a path~$P_c = (u_0,\ldots,u_h)$ with~$u_0=s_\mu$ and~$u_h=t_\mu$. Each edge~$(u_i,u_{i+1})$ of~$P_c$ corresponds to a node~$\nu_i$ of~$\mathcal{T}$, augmented or not, for which a stretchable drawing of~$\overline{G}_{\nu_i}$ has already been computed.
Also the poles $u_i$ and~$u_{i+1}$ of~$\nu_i$ are real vertices.  
We construct a stretchable drawing of~$\overline{G}_\mu$  
similarly to~\cite{AngeliniBLM19}. Unlike~\cite{AngeliniBLM19}, however, a chip~$C_{\nu_i}$ may not have pins at both bottom corners, and some drawings may need to be flipped to preserve the~embedding.

Let $\rho_t$ and $\rho'_t$ be the first top rays of a vertex encountered when traversing clockwise and counterclockwise, respectively, starting from the horizontal ray (see \cref{fig:s_node_rays}).
Let~$\rho_b$ and~$\rho'_b$ be the corresponding bottom rays. 
We place~$u_1, \ldots, u_{h-1}$ in this order along a~horizontal line $L$ and process the virtual edges and their corresponding nodes from left to right. To preserve the embedding of~$G$, we may need to flip~$C_{\nu_i}$ vertically to ensure that the bottom-to-top order of its pins 
matches the clockwise and counterclockwise edge orderings around~$u_{i-1}$ and~$u_i$. For~$i = 2, \ldots, h-1$, chip~$C_{\nu_i}$ (\cref{fig:s_node_example}) is uniformly scaled down and positioned horizontally between~$u_{i-1}$ and~$u_i$ such that it does not intersect~$\rho'_t$ and~$\rho'_b$~emanating from~$u_{i-1}$ and~$\rho_t$ and~$\rho_b$ emanating from~$u_{i}$.
Additionally,~$C_{\nu_i}$ is aligned vertically such that, if $(u_{i-1}, u_i)$ is a real edge, then the bottom boundary of the chip is slightly above~$L$ such that $(u_{i-1},u_i)$ is drawn as a horizontal segment; otherwise, the~$y$-coordinate of its bottommost left pin coincides with~$u_{i-1}$. 

The boundary chips~$C_{\nu_1}$ and~$C_{\nu_h}$ are placed analogously. Chip~$C_{\nu_1}$ ($C_{\nu_h}$) is positioned to the left (right) of~$u_1$ ($u_{h-1}$) without intersecting the rays~$\rho_b$ and~$\rho_t$ ($\rho'_b$ and~$\rho'_t$) emanating from~$u_1$ ($u_{h-1}$). If~$(s_\mu, u_1)$ ($(u_{h-1}, t_\mu)$) is not a real~edge, its bottommost pin on the left side is aligned with line~$L$; otherwise, its bottom boundary is positioned slightly above~$L$. If~$\nu_1$ (resp.~$\nu_h$) is a~$Q$-node, then the edge~$(s_\mu, u_1)$ (resp.~$(u_{h-1}, t_\mu)$) belongs to~$G$, and the pin corresponding to this edge is the only pin of~$s_\mu$ (resp.~$t_\mu$).

\begin{figure}
    \centering
    \subcaptionbox{\nolinenumbers \label{fig:s_node_rays}}{\includegraphics[page=3]{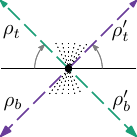}}
     \hfil
    \subcaptionbox{\nolinenumbers \label{fig:s_node_example}}{\includegraphics[page=4]{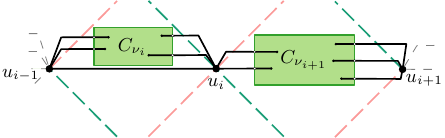}}     
    \caption{Illustration of \textbf{(a)} rays $\rho_t$, $\rho'_t$, $\rho_b$, $\rho'_b$, \textbf{(b)} the placement of $C_{\nu_i}$ between $u_{i-1}$ and $u_i$.}
    \label{fig:s_node}
\end{figure}

We now draw all edges incident to each~$u_i$, $i = 1, \ldots, h-1$. For $i = 2, \ldots, h-1$, let~$d_i^t$ and~$d_i^b$ be the number of neighbors of~$u_i$ in~$\overline{G}_{\nu_i}$ above and below~$L$, respectively. We draw these edges using the first~$d_i^t$ top rays in clockwise order and the first~$d_i^b$ bottom rays in counterclockwise order, applying \cref{lem:poles-addition} twice: once for the bottom ports and once for the top ones. If a neighbor of~$u_i$ in $C_{v_i}$ is on~$L$, it is connected to~$u_i$ by a horizontal segment. Symmetrically, we draw the edges connecting $u_i$ to its~neighbors in $\overline{G}_{\nu_{i+1}}$.
By construction, the right horizontal port of~$u_i$ is always occupied, and at most~$\Delta-1$ top ports are used.
Finally, $C_\mu$ is the smallest rectangle that surrounds the current drawing, which is stretchable as the drawings of~$\overline{G}_{\nu_1}$ and~$\overline{G}_{\nu_h}$ are stretchable.

\medskip\noindent\textbf{$\mu$ is an $R$-node:}
We follow closely the approach developed in~\cite{AngeliniBLM19}.
If $\mu$ has been augmented, we apply a preprocessing. Recall that in this case, one or both of the poles of $\mu$ were dummy~vertices, and in the augmentation step, we attached to each of them the corresponding satellite~$Q$-node.
Assume first that both poles $s_\mu$ and $t_\mu$ were dummy vertices (\cref{fig:r_node_1}). 
Let $\sigma_\mu$ and $\tau_\mu$ be the poles of the augmented node $\mu$, so $(s_\mu,\sigma_\mu)$ and $(t_\mu,\tau_\mu)$ are edges of the skeleton of $\mu$. 
Since the skeleton of $\mu$ is not 3-connected, we augment $G_\mu$
by adding edges $(\sigma_\mu,\tau_\mu)$, $(\sigma_\mu, v)$ and $(\tau_\mu, v)$, where $v$ is adjacent to $s_\mu$ but not adjacent to $\sigma_\mu$ in~$H$. 

If only one of the poles of $\mu$, say $s_\mu$, were a dummy vertex and $t_\mu$ is a real vertex (see \cref{fig:r_node_2}), then let $\sigma_\mu$ be the pole of the augmented node $\mu$ adjacent to $s_\mu$, such that $(s_\mu, \sigma_\mu)$ is an edge in the skeleton of $\mu$.
Again, as $s_\mu$ and $t_\mu$ form a separation pair, we must augment $G_\mu$ to achieve triconnectivity.
First, we add the edge $(\sigma_\mu,t_\mu)$. 
Then, we add the edge $(\sigma_\mu, v)$, where $v$ is adjacent to $s_\mu$ but not to $\sigma_\mu$ in the original graph $H$.

This completes the preprocessing step for $\mu$ and ensures that the skeleton of $\mu$ is triconnected. Note that the auxiliary edges are used to compute the canonical ordering of $G_\mu$ required by the algorithm in \cref{sec:triconnected} such that $P_0 = \{\sigma_\mu, \tau_\mu\}$ (resp.\ $P_0 = \{\sigma_\mu, t_\mu\}$ if only~$s_\mu$ was a dummy vertex) and $P_m = \{v\}$.
Although these edges increase the degree of certain vertices, at most one additional slope outside of~$S$ is needed, and only temporarily during the construction, as we argue next.
Before the augmentation,~$\sigma_\mu$ and~$\tau_\mu$ have degree~$1$ in~$G_\mu$, so their degree is at most~$3$ afterwards.
Vertex~$v$, however, may have degree at most~$\Delta+2$.
Since~$P_m = \{v\}$, all edges incident to~$v$ are incoming and, by the port assignment of \cref{sec:triconnected}, use bottom or horizontal ports of~$v$, of which there are only~$\Delta+1$.
We therefore add one auxiliary slope not in~$S$, chosen angularly between the horizontal slope and its neighboring slope in~$S$, which provides one additional bottom port at~$v$.
The auxiliary edge~$(\sigma_\mu,v)$ uses the horizontal left port of~$v$, the auxiliary edge~$(\tau_\mu,v)$ uses the bottom port of the auxiliary slope, and the at most~$\Delta$ real edges incident to~$v$ use the remaining ports as in \cref{sec:triconnected}; if the degree of~$v$ in~$G_\mu$ is at most~$\Delta-1$, no auxiliary slope is needed.
Since every auxiliary edge is incident to~$\sigma_\mu$ or~$\tau_\mu$, removing the poles to define the chip~$C_\mu$ removes all auxiliary edges, and in particular every segment drawn with the auxiliary slope, without creating pins for them.
Thus, every segment of the chip~$C_\mu$ has a slope in~$S$, so the auxiliary slope exists only during the construction.
\begin{figure}
    \centering
    \subcaptionbox{\nolinenumbers \label{fig:r_node_1}}{\includegraphics[page=1]{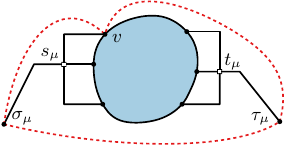}}
     \hfil
    \subcaptionbox{\nolinenumbers \label{fig:r_node_2}}{\includegraphics[page=2]{figures/r_node.pdf}}     
    \caption{The augmented $R$-node when \textbf{(a)} both poles are dummy, \textbf{(b)} one pole ($s_\mu$) is dummy.}
    \label{fig:r_node}
\end{figure}

We compute a stretchable drawing of $\overline{G}_\mu$ by applying the algorithm from \cref{sec:triconnected} to the entire pertinent graph of $\mu$. We modify this algorithm to account for the fact that each virtual edge $(u, v)$ actually corresponds to the subgraph $G_\nu$ of a child node $\nu$ with poles $s_\nu = u$ and $t_\nu = v$. 
Afterwards, if a child~$\nu$ of~$\mu$ has been reduced at a dummy pole~$s_\nu$ with unique neighbor~$u_\nu$ in its (unreduced) pertinent graph, we replace the virtual edge~$(s_\nu,t_\nu)$ in the skeleton of~$\mu$ by the path consisting of the real edge~$(s_\nu,u_\nu)$ and a virtual edge~$(u_\nu,t_\nu)$ that corresponds to the reduced pertinent graph~$G_\nu$; if~$\nu$ has been reduced at both poles, we replace both ends of the virtual edge accordingly. Note that the graph processed by the algorithm of \cref{sec:triconnected} is then not the planarization of a $1$-planar graph, as it may contain edges between two dummy vertices. Thus, we have to extend the invariants and the algorithm of \cref{sec:triconnected} to handle this.

\begin{enumerate}[label=I.\arabic*$^*$]
    \item\label{inv:new-cut}There is a cut through every contour edge, except the last edge of each dummy vertex whose neighbor is a real vertex.
    \setcounter{enumi}{3}
    \item\label{inv:new-dummy} For each dummy vertex $v_d$ in $G_k$ due to the crossing of $(v_i,v_{i'})$ and $(v_j,v_{j'})$ such that $v_{j'}$ is its last neighbor, the following hold.
        \begin{enumerate}[label=\alph*),ref=I.4$^*$(\alph*)]
            \item If $v_i\in G_k$ and~$v_i$ is real, then $(v_d,v_i)$ is drawn with at most one bend in $\Gamma_k$;
            \item If $v_{i'}\in G_k$ and~$v_{i'}$ is real, then $(v_d,v_{i'})$ is drawn with at most one bend in $\Gamma_k$;
            \item If $v_{j'}\in G_k$,~$v_{j'}$ is real, and $(v_d,v_{j'})$ is not a virtual edge, then $(v_d,v_{j'})$ is drawn without bends in $\Gamma_k$.
            \item\label{inv*:last-edge-virtual} If $v_{j'}\in G_k$,~$v_{j'}$ is real, and $(v_d,v_{j'})$ a virtual edge, then $(v_d,v_{j'})$ is drawn with at most two bends and with one horizontal segment in $\Gamma_k$.
            \item\label{inv*:dummy-last} If $v_j\in G_k$,~$v_j$ is dummy, and $v_d$ is not the last neighbor of $v_j$, then $(v_j,v_d)$ is drawn with at most three bends and with two horizontal segments.
            \item\label{inv*:dummy-not-last} Every other edge between~$v_d$ and another dummy vertex is drawn with at most two bends and at least one horizontal segment.
        \end{enumerate}
\end{enumerate}

In the port assignment phase, the port used at a dummy vertex~$v_d$ for its last edge $(v_d,v_{j'})$ was chosen as the port opposite of the port~$\pi$ used at $v_{j'}$. If~$v_{j'}$ is also dummy, then~$\pi$ is not yet defined; instead, we may choose~$\pi$ as an arbitrary bottom port at~$v_{j'}$.
When processing a virtual edge~$(v_h, v_g)$,
we do not reserve a single port for it at its two endvertices~$v_h$ and~$v_g$. Instead, we reserve as many consecutive ports as the degree of~$v_h$ and~$v_g$ in~$G_\nu$, respectively, plus one in the case in which~$(v_h,v_g)$ is a real edge, starting from the one that the algorithm described in \cref{sec:triconnected} would assign.
If $G_\nu$ has been reduced at $v_g$ or $v_h$, then we reserve a single port at that vertex.
An exception to this rule occurs when~$(v_h, v_g)$ is a virtual edge corresponding to a child $\nu$ of $\mu$ that belongs to a chain but is not a real edge of $H$. Rather than reserving~$d$ consecutive ports at~$v_g$ (where~$d$ is the degree of~$v_g$ in~$G_\nu$), we reserve the horizontal port, the first~$d$ top and the first~$d$ bottom ports in clockwise and counterclockwise order around $v_g$, respectively.
The right horizontal port of every vertex on a chain is always occupied.
This approach leverages a key property of the algorithm in \cref{sec:triconnected}: for a vertex $v_i$ on a chain, the left and right horizontal ports are assigned to its incident chain edges. Since $v_i$ has at most $\Delta-2$ outgoing edges but $\Delta-1$ available top ports, there is always at least one unassigned top~port.

\begin{figure}
    \centering
    \subcaptionbox{\nolinenumbers \label{fig:chip-real}}{\includegraphics[page=1]{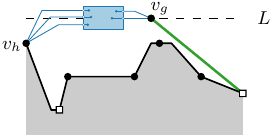}}
     \hfil
    \subcaptionbox{\nolinenumbers \label{fig:chip-chain}}{\includegraphics[page=2]{figures/chips.pdf}}     
     \hfil
    \subcaptionbox{\nolinenumbers \label{fig:chip-high}}{\includegraphics[page=3]{figures/chips.pdf}} 
    \caption{Attaching a chip on a virtual edge \textbf{(a)} incident to a degree-2 real singleton, \textbf{(b)} chain, \textbf{(c)} high-degree singleton.}
    \label{fig:chips}
\end{figure}

We now describe how to adjust the drawing algorithm of \cref{sec:triconnected} to handle edges between two dummy vertices $v_h$ and $v_g$; see \cref{fig:dummy-dummy}. If $v_h\prec v_g$, then we draw the edge $(v_h,v_g)$ as if $v_h$ was real; however, if $v_g$ is not the last neighbor of $v_h$, we start with a (very short) horizontal segment at~$v_h$. Otherwise, $v_h$ and $v_g$ are consecutive on a chain. We draw the edge locally at~$v_h$ and~$v_g$ as described in \cref{sec:triconnected}. If $(v_h,v_g)$ is critical for both $v_h$ and $v_g$, this might result in an edge with two bends. This satisfies invariants~\ref{inv*:dummy-last} and~\ref{inv*:dummy-not-last}.
Finally, if there is a dummy vertex~$v_h$ whose final edge $(v_h,v_g)$ is a virtual edge and~$v_g$ is a real vertex, then we also draw $(v_h,v_g)$ as if $v_h$ was a real vertex to satisfy invariant~\ref{inv*:last-edge-virtual}.

Let~$\nu$ be a reduced child of~$\mu$ with poles $(u_\nu,v_\nu)$, and let $(s_\nu,t_\nu)$ be the virtual edge corresponding to $\nu$. Then, either one or both of $s_\nu$ and $t_\nu$ is a dummy vertex, so we have to place one or two subdivision vertices (namely, $u_\nu$ and/or $v_\nu$) on $(s_\nu,t_\nu)$; see \cref{fig:dummy-dummy}.
By \ref{inv:new-dummy}, the edge $(s_\nu,t_\nu)$ contains at least one horizontal segment; we call the last horizontal segment in the direction from $s_\nu$ to $t_\nu$ the \emph{host}.
If $u_\nu\neq s_\nu$, we place it on the host as follows:
if $s_\nu$ and $t_\nu$ are consecutive on a chain, then we place $u_\nu$ on an arbitrary interior point of the host; otherwise, we place $u_\nu$ on the bend point of the host that is closer to $s_\nu$.
If $u_\nu=s_\nu$, then we have $v_\nu\neq t_\nu$, and we place $v_\nu$ in this position instead.
If both $u_\nu\neq s_\nu$ and $v_\nu\neq t_\nu$, then we place $v_\nu$ on an interior point of the host, between $u_\nu$ and $t_\nu$.
Note that this placement satisfies the original invariant \ref{inv:dummy}; in particular, the last edge of every dummy vertex is drawn without bends, the critical edge is drawn with two bends, and the other two edges are drawn with one bend each. 
Thus, the resulting drawing gives a 2-bend 1-planar drawing after removing the dummy vertices.

\begin{figure}
    \centering
    \includegraphics[page=9]{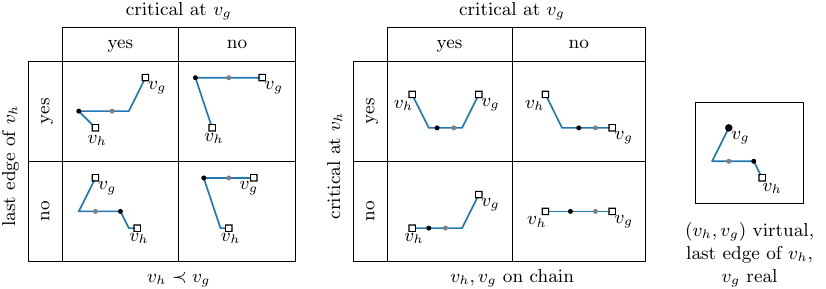}
    \caption{Drawing a dummy-dummy edge or a virtual last dummy-real edge and placing one (black) or two (gray) subdivision vertices on it.}
    \label{fig:dummy-dummy}
\end{figure}

In the following, we show how to replace the virtual edges of non-reduced nodes. 
In the drawing phase, when the virtual edge $(v_h, v_g)$ is encountered, we add the stretchable drawing of~$G_\nu$, together with the edge $(v_h, v_g)$, if this is a real edge. Since each edge connecting a pair of real vertices contains a horizontal segment and since each chip connects such vertices and is in turn stretchable, the introduction of the stretchable drawings at each virtual edge can be done similarly to~\cite{AngeliniBLM19}; see \cref{fig:chips}. 
Since the case in which~$(v_h,v_g)$ is an edge of a degree-$2$ real singleton is rather easy (\cref{fig:chip-real}), we focus on the case in which~$(v_h,v_g)$ is either an edge of a chain (\cref{fig:chip-chain}) or is incident to a high-degree real singleton (\cref{fig:chip-high}).

If $(v_h, v_g)$ is a chain edge, we position the chip between $v_h$ and $v_g$ following the same logic used for the $S$-node case between vertices $u_{i-1}$ and $u_i$.

Suppose, now, that $(v_h,v_g)$ is an edge incident to a high-degree real singleton $v_g$. 
Let $d_h$ and $d_g$ be the degrees of $v_h$ and $v_g$, respectively, in~$G_\nu$, and assume that the edge $(v_h,v_g)$ is a real edge; the case in which this edge is not real is simpler.  
In this case, our algorithm has reserved $d_h+1$ ports at $v_h$ and $d_g+1$ ports at $v_g$ for the introduction of the stretchable drawing of~$\overline{G}_\nu$. Let $\rho_h$ ($\rho_g$) be the last reserved such port incident to $v_h$ ($v_g$) that is encountered in a clockwise (counterclockwise) traversal of the ports incident to $v_h$ ($v_g$) starting from the horizontal left one. 
In the procedure with the two sweeps presented in \cref{ssec:real-highdeg}, we may need to further stretch the drawing to guarantee that the crossing point of $\rho_h$ with $L$ is to the left of the one with $\rho_g$. This enables the stretchable drawing of $\overline{G}_\nu$ to be placed (after a possible scaling) between these two crossing points. Consequently, the edges connecting poles $v_h$ and $v_g$ to the pins of the drawing of $\overline{G}_\nu$ can be drawn without crossings by applying \cref{lem:poles-addition}.
Since the edge connecting each pin of $C_{\nu}$
to its endvertex in~$\overline{G}_\nu$ is drawn with at most one bend, 
the edges from each of the poles of $v_h$ and $v_g$ are drawn with at most two bends each. 

It remains to describe the case where $\nu$ is a reduced node. As described above (see also \cref{fig:dummy-dummy}, we have placed its (real) poles on a horizontal segment. Hence, we can position the chip between the poles exactly as in the chain case; see \cref{fig:chip-chain}.

Finally, we handle the case that no separation pair with two real vertices exists. 
Let~$\langle s,t\rangle$ be a separation pair; w.l.o.g.\ $s$ is dummy. By Property~\ref{prp:poles} of~\cref{lem:planarization}, $s$ has exactly one neighbor $s'$ in one of the connected components of $G\setminus\{s,t\}$; thus, $\langle s',t\rangle$ is also a separation pair. Note that $s'$ is real and $s'\neq t$. 
If~$t$ is real, then $\langle s', t \rangle$ is a separation pair consisting of two real vertices; hence, $t$ is a dummy vertex. In this case, pick the (real) neighbor $t'$ of $t$ analogously. 
If $s'\neq t'$, then $\langle s',t'\rangle$ is a separation pair consisting of two real vertices; hence, $s'=t'$. Since $s'=t'$ is a vertex of degree $2$, we split it into two copies, say $v$ and $v'$, both of which we assume to be real. We connect $v$ to one of the neighbors of $s'=t'$, $v'$ to the other one and we make $v$ and $v'$ adjacent. This allows us to root $T$ at the $Q$-node corresponding to $(v,v')$. Once the traversal of $T$ is completed, we remove both $v$ and $v'$ and we draw the original graph as depicted in \cref{fig:exception}.

\begin{figure}
    \centering
    \includegraphics[page=8]{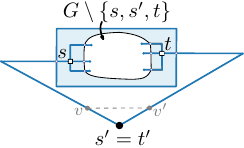}
    \caption{Handling the case that there is no separation pair of real vertices.}
    \label{fig:exception}
\end{figure}

We conclude this section by noting that once the stretchable drawing of $\overline{G}_\mu$ is obtained, the invariant property that one may use one more bend to connect each pin to the pole of $\mu$ is maintained. This is clearly the case when $\mu$ is an $S$- or a $P$-node (by construction). When~$\mu$ is an $R$-node, then its poles are $v_1$ and $v_2$ in the canonical order. Since in the drawing of~$G_\mu$ that we constructed, each edge incident to these two vertices has a horizontal segment which is incident to neither $v_1$ nor $v_2$, the chip $C_\mu$ can be defined in a way to guarantee the invariant property by subdividing each of these horizontal segments.

\begin{theorem}\label{thm:biconnectedcase}
For any $\Delta \geq 4$, every set $S$ of $\Delta$ slopes is universal for $2$-bend $1$-planar drawings of biconnected $1$-planar graphs with maximum degree $\Delta$. Also, for any such graph on $n$ vertices, a $2$-bend $1$-planar drawing on $S$ can be computed in $O(n)$ time if a $1$-planar embedding is specified as part of the input.    
\end{theorem}

\section{Conclusions}
\label{sec:conclusions}

In this work, we studied the $2$-bend slope number of $1$-planar graphs, thereby providing the first investigation of the problem for a beyond-planar graph class of
arbitrary degree.
We showed that every biconnected $1$-planar graph of maximum degree $\Delta$ admits a $2$-bend drawing using any given set of $\Delta$ slopes. 
We strongly believe that our result extends to simply connected $1$-planar graphs using a similar approach as the one in \cite{AngeliniBLM19} by ensuring that the biconnected components are drawn with enough consecutive ports available at the cut vertices.
We also believe that our algorithm can be extended to draw $k$-planar graphs with $\Delta$ slopes and $2k$ bends per edge, for $k>1$.

Several directions, however, remain unexplored.
A key open problem is whether every $1$-planar graph admits a $1$-bend drawing where the slope number is bounded by a function of $\Delta$.
Another important question concerns the area requirement. Although our algorithm uses a bounded number of slopes and at most two bends per edge, it is unclear whether the output drawings fit on polynomial-area grids.
A different direction is to study trade-offs between the number of slopes and the number of bends per edge, in particular, whether similar results can be obtained with fewer bends (even for subclasses of $1$-planar graphs). 
Finally, as a first step towards other beyond-planar graph classes, optimal 2-planar~\cite{DBLP:conf/compgeom/Bekos0R17} and framed graphs~\cite{DBLP:conf/compgeom/BekosLGGMR20} appear to be particularly promising.

\bibliographystyle{plainurl}
\bibliography{bibliography}

@article{KPP2013,
author = {Keszegh, Bal\'{a}zs and Pach, J\'{a}nos and P\'{a}lv\"{o}lgyi, D\"{o}m\"{o}t\"{o}r},
title = {Drawing Planar Graphs of Bounded Degree with Few Slopes},
journal = {SIAM J. Discr. Math.},
volume = {27},
number = {2},
pages = {1171-1183},
year = {2013},
doi = {10.1137/100815001}
}

@InProceedings{KPP2010GD,
author="Keszegh, Bal{\'a}zs
and Pach, J{\'a}nos
and P{\'a}lv{\"o}lgyi, D{\"o}m{\"o}t{\"o}r",
editor="Brandes, Ulrik
and Cornelsen, Sabine",
doi="10.1007/978-3-642-18469-7_27",
title="Drawing Planar Graphs of Bounded Degree with Few Slopes",
booktitle="Graph Drawing (GD 2011)",
year="2011",
publisher="Springer",
series="LNCS",
pages="293--304"
}

@article{BiedlKant1998,
title = {A better heuristic for orthogonal graph drawings},
journal = {Comput. Geom.},
volume = {9},
number = {3},
pages = {159-180},
year = {1998},
issn = {0925-7721},
doi = {https://doi.org/10.1016/S0925-7721(97)00026-6},
author = {Therese Biedl and Goos Kant}
}

@article{FraysseixPP90,
  author       = {Hubert de Fraysseix and
                  J{\'{a}}nos Pach and
                  Richard Pollack},
  title        = {How to draw a planar graph on a grid},
  journal      = {Comb.},
  volume       = {10},
  number       = {1},
  pages        = {41--51},
  year         = {1990},
  doi          = {10.1007/BF02122694}
}

@article{AngeliniBLM19,
  author       = {Patrizio Angelini and
                  Michael A. Bekos and
                  Giuseppe Liotta and
                  Fabrizio Montecchiani},
  title        = {Universal Slope Sets for 1-Bend Planar Drawings},
  journal      = {Algorithmica},
  volume       = {81},
  number       = {6},
  pages        = {2527--2556},
  year         = {2019},
  doi          = {10.1007/S00453-018-00542-9}
}

@book{2013gd,
  editor       = {Roberto Tamassia},
  title        = {Handbook on Graph Drawing and Visualization},
  publisher    = {Chapman and Hall/CRC},
  year         = {2013},
  url          = {https://www.crcpress.com/Handbook-of-Graph-Drawing-and-Visualization/Tamassia/9781584884125},
  isbn         = {978-1-5848-8412-5}
}

@book{Juenger04,
  editor       = {Michael J{\"{u}}nger and
                  Petra Mutzel},
  title        = {Graph Drawing Software},
  publisher    = {Springer},
  year         = {2004},
  doi          = {10.1007/978-3-642-18638-7},
  isbn         = {978-3-642-62214-4}
}

@article{Kant96,
  author       = {Goos Kant},
  title        = {Drawing Planar Graphs Using the Canonical Ordering},
  journal      = {Algorithmica},
  volume       = {16},
  number       = {1},
  pages        = {4--32},
  year         = {1996},
  doi          = {10.1007/BF02086606}
}

@Book{Felsner04,
	author    = {Stefan Felsner},
	title     = {Geometric Graphs and Arrangements},
	publisher = {Vieweg},
	year      = {2004},
	series    = {Advanced Lectures in Mathematics},
	doi		  = {10.1007/978-3-322-80303-0}
}

@inproceedings{Schnyder90,
  author    = {Walter Schnyder},
  title     = {Embedding Planar Graphs on the Grid},
  booktitle = {Discrete Algorithms ({SODA} 1990)},
  pages     = {138--148},
  publisher = {{ACM-SIAM}},
  year      = {1990},
  doi       = {10.5555/320176.320191}
}

@article{WadeChu1994,
	author       = {Greg A. Wade and
                  Jiang{-}Hsing Chu},
    title        = {Drawability of Complete Graphs Using a Minimal Slope Set},
    journal = {The Comput. J.},
    volume       = {37},
    number       = {2},
    pages        = {139--142},
    year         = {1994},
    doi          = {10.1093/COMJNL/37.2.139}
}

@article{Tamassia1987,
	author = {Tamassia, Roberto},
	title = {On Embedding a Graph in the Grid with the Minimum Number of Bends},
	journal = {SIAM J. Comput.},
	volume = {16},
	number = {3},
	pages = {421-444},
	year = {1987},
	doi = {10.1137/0216030}
}

@article{BekosGKK15, 
	title={Planar Octilinear Drawings with One Bend Per Edge}, 
	volume={19},
	DOI={10.7155/jgaa.00369},
	number={2}, 
	journal={J. Graph Algorithms Appl.}, 
	author={Bekos, Michael A. and Gronemann, Martin and Kaufmann, Michael and Krug, Robert}, 
	year={2015}, 
	pages={657--680} 
}

@inproceedings{BekosGKK14,
  author       = {Michael A. Bekos and
                  Martin Gronemann and
                  Michael Kaufmann and
                  Robert Krug},
  editor       = {Christian A. Duncan and
                  Antonios Symvonis},
  title        = {Planar Octilinear Drawings with One Bend Per Edge},
  booktitle    = {Graph Drawing ({GD} 2014)},
  series       = {LNCS},
  volume       = {8871},
  pages        = {331--342},
  publisher    = {Springer},
  year         = {2014},
  doi          = {10.1007/978-3-662-45803-7_28},
}

@inproceedings{DBLP:conf/compgeom/Bekos0R17,
  author       = {Michael A. Bekos and
                  Michael Kaufmann and
                  Chrysanthi N. Raftopoulou},
  editor       = {Boris Aronov and
                  Matthew J. Katz},
  title        = {On Optimal 2- and 3-Planar Graphs},
  booktitle    = {Computational Geometry (SoCG 2017)},
  series       = {LIPIcs},
  pages        = {16:1--16:16},
  publisher    = {Schloss Dagstuhl - Leibniz-Zentrum f{\"{u}}r Informatik},
  year         = {2017},
  doi          = {10.4230/LIPICS.SOCG.2017.16},
}

@inproceedings{DBLP:conf/compgeom/BekosLGGMR20,
  author       = {Michael A. Bekos and
                  Giordano {Da Lozzo} and
                  Svenja Griesbach and
                  Martin Gronemann and
                  Fabrizio Montecchiani and
                  Chrysanthi N. Raftopoulou},
  editor       = {Sergio Cabello and
                  Danny Z. Chen},
  title        = {Book Embeddings of Nonplanar Graphs with Small Faces in Few Pages},
  booktitle    = {Computational Geometry (SoCG 2020)},
  series       = {LIPIcs},
  pages        = {16:1--16:17},
  publisher    = {Schloss Dagstuhl - Leibniz-Zentrum f{\"{u}}r Informatik},
  year         = {2020},
  doi          = {10.4230/LIPICS.SOCG.2020.16},
}

@InProceedings{DiGiacomoLM2014,
	author={{Di Giacomo}, Emilio
	and Liotta, Giuseppe
	and Montecchiani, Fabrizio},
	editor={Pardo, Alberto
	and Viola, Alfredo},
	title={The Planar Slope Number of Subcubic Graphs},
	booktitle={Theoretical Informatics (LATIN 2014)},
	year={2014},
	publisher={Springer},
	pages={132--143},
    doi = {10.1007/978-3-642-54423-1_12},
}

@article{GiacomoLM18,
	author       = {{Di Giacomo}, Emilio
	and Liotta, Giuseppe
	and Montecchiani, Fabrizio},
	title        = {Drawing subcubic planar graphs with four slopes and optimal angular
	resolution},
	journal      = {Theor. Comput. Sci.},
	volume       = {714},
	pages        = {51--73},
	year         = {2018},
	doi          = {10.1016/J.TCS.2017.12.004}
}

@article{GiacomoLM15,
	author       = {{Di Giacomo}, Emilio
	and Liotta, Giuseppe
	and Montecchiani, Fabrizio},
	title        = {Drawing Outer 1-planar Graphs with Few Slopes},
	journal      = {J. Graph Algorithms Appl.},
	volume       = {19},
	number       = {2},
	pages        = {707--741},
	year         = {2015},
	doi          = {10.7155/JGAA.00376}
}

@inproceedings{GiacomoLM14,
	author       = {{Di Giacomo}, Emilio
	and Liotta, Giuseppe
	and Montecchiani, Fabrizio},
	editor       = {Christian A. Duncan and
	Antonios Symvonis},
	title        = {Drawing Outer 1-planar Graphs with Few Slopes},
	booktitle    = {Graph Drawing ({GD} 2014)},
	series       = {LNCS},
	volume       = {8871},
	pages        = {174--185},
	publisher    = {Springer},
	year         = {2014},
	doi          = {10.1007/978-3-662-45803-7_15}
}

@article{GiacomoLM20,
	author       = {{Di Giacomo}, Emilio
	and Liotta, Giuseppe
	and Montecchiani, Fabrizio},
	title        = {1-bend upward planar slope number of {SP}-digraphs},
	journal      = {Comput. Geom.},
	volume       = {90},
	pages        = {101628},
	year         = {2020},
	doi          = {10.1016/J.COMGEO.2020.101628}
}

@inproceedings{GiacomoLM16,
author       = {{Di Giacomo}, Emilio
and Liotta, Giuseppe
and Montecchiani, Fabrizio},
editor       = {Yifan Hu and
Martin N{\"{o}}llenburg},
title        = {1-Bend Upward Planar Drawings of SP-Digraphs},
booktitle    = {Graph Drawing and Network Visualization ({GD} 2016)},
series       = {LNCS},
volume       = {9801},
pages        = {123--130},
publisher    = {Springer},
year         = {2016},
doi          = {10.1007/978-3-319-50106-2_10},
}

@inproceedings{ChaplickLGLM21,
  author       = {Steven Chaplick and
                  Giordano Da Lozzo and
                  Emilio Di Giacomo and
                  Giuseppe Liotta and
                  Fabrizio Montecchiani},
  editor       = {Anna Lubiw and
                  Mohammad R. Salavatipour},
  title        = {Planar Drawings with Few Slopes of {H}alin Graphs and Nested Pseudotrees},
  booktitle    = {Algorithms and Data Structures (WADS 2021)},
  series       = {LNCS},
  volume       = {12808},
  pages        = {271--285},
  publisher    = {Springer},
  year         = {2021},
  doi          = {10.1007/978-3-030-83508-8_20},
}

@article{ChaplickLGLM24,
	author       = {Steven Chaplick and
	Giordano {Da Lozzo} and
	Emilio {Di Giacomo} and
	Giuseppe Liotta and
	Fabrizio Montecchiani},
	title        = {Planar Drawings with Few Slopes of {H}alin Graphs and Nested Pseudotrees},
	journal      = {Algorithmica},
	volume       = {86},
	number       = {8},
	pages        = {2413--2447},
	year         = {2024},
	doi          = {10.1007/S00453-024-01230-7}
}

@article{KindermannMSS21,
author       = {Philipp Kindermann and
Fabrizio Montecchiani and
Lena Schlipf and
Andr{\'{e}} Schulz},
title        = {Drawing Subcubic 1-Planar Graphs with Few Bends, Few Slopes, and Large
Angles},
journal      = {J. Graph Algorithms Appl.},
volume       = {25},
number       = {1},
pages        = {1--28},
year         = {2021},
doi          = {10.7155/JGAA.00547}
}

@inproceedings{KindermannMSS18,
author       = {Philipp Kindermann and
Fabrizio Montecchiani and
Lena Schlipf and
Andr{\'{e}} Schulz},
editor       = {Therese Biedl and
Andreas Kerren},
title        = {Drawing Subcubic 1-Planar Graphs with Few Bends, Few Slopes, and Large
Angles},
booktitle    = {Graph Drawing and Network Visualization (GD 2018)},
series       = {LNCS},
volume       = {11282},
pages        = {152--166},
publisher    = {Springer},
year         = {2018},
doi          = {10.1007/978-3-030-04414-5_11}
}

@article{LenhartLMN23,
author       = {William J. Lenhart and
Giuseppe Liotta and
Debajyoti Mondal and
Rahnuma Islam Nishat},
title        = {Drawing Partial 2-Trees with Few Slopes},
journal      = {Algorithmica},
volume       = {85},
number       = {5},
pages        = {1156--1175},
year         = {2023},
doi          = {10.1007/S00453-022-01065-0}
}

@article{MukkamalaSzegedy2009,
	author       = {Padmini Mukkamala and
	Mario Szegedy},
	title        = {Geometric representation of cubic graphs with four directions},
	journal      = {Comput. Geom.},
	volume       = {42},
	number       = {9},
	pages        = {842--851},
	year         = {2009},
	doi          = {10.1016/J.COMGEO.2009.01.005}
}

@inproceedings{KeszeghPPT06,
	author       = {Bal{\'{a}}zs Keszegh and
	J{\'{a}}nos Pach and
	D{\"{o}}m{\"{o}}t{\"{o}}r P{\'{a}}lv{\"{o}}lgyi and
	G{\'{e}}za T{\'{o}}th},
	editor       = {Michael Kaufmann and
	Dorothea Wagner},
	title        = {Drawing Cubic Graphs with at Most Five Slopes},
	booktitle    = {Graph Drawing and Network Visualization ({GD} 2006)},
	series       = {LNCS},
	volume       = {4372},
	pages        = {114--125},
	publisher    = {Springer},
	year         = {2006},
	doi          = {10.1007/978-3-540-70904-6_13}
}

@article{KeszeghPPT08,
	author       = {Bal{\'{a}}zs Keszegh and
	J{\'{a}}nos Pach and
	D{\"{o}}m{\"{o}}t{\"{o}}r P{\'{a}}lv{\"{o}}lgyi and
	G{\'{e}}za T{\'{o}}th},
	title        = {Drawing cubic graphs with at most five slopes},
	journal      = {Comput. Geom.},
	volume       = {40},
	number       = {2},
	pages        = {138--147},
	year         = {2008},
	doi          = {10.1016/J.COMGEO.2007.05.003}
}

@book{BattistaETT99,
  author       = {Giuseppe Di Battista and
                  Peter Eades and
                  Roberto Tamassia and
                  Ioannis G. Tollis},
  title        = {Graph Drawing: Algorithms for the Visualization of Graphs},
  publisher    = {Prentice-Hall},
  year         = {1999},
  isbn         = {0-13-301615-3},
  bibsource    = {dblp computer science bibliography, https://dblp.org}
}

@inproceedings{MukkamalaP11,
	author       = {Padmini Mukkamala and
	D{\"{o}}m{\"{o}}t{\"{o}}r P{\'{a}}lv{\"{o}}lgyi},
	editor       = {Marc J. van Kreveld and
	Bettina Speckmann},
	title        = {Drawing Cubic Graphs with the Four Basic Slopes},
	booktitle    = {Graph Drawing (GD 2011)},
	series       = {LNCS},
	volume       = {7034},
	pages        = {254--265},
	publisher    = {Springer},
	year         = {2011},
	doi          = {10.1007/978-3-642-25878-7_25}
}

@article{DujmovicSW07,
	author       = {Vida Dujmovic and
	Matthew Suderman and
	David R. Wood},
	title        = {Graph drawings with few slopes},
	journal      = {Comput. Geom.},
	volume       = {38},
	number       = {3},
	pages        = {181--193},
	year         = {2007},
	doi          = {10.1016/J.COMGEO.2006.08.002}
}

@article{PachP06,
	author       = {J{\'{a}}nos Pach and
	D{\"{o}}m{\"{o}}t{\"{o}}r P{\'{a}}lv{\"{o}}lgyi},
	title        = {Bounded-Degree Graphs can have Arbitrarily Large Slope Numbers},
	journal      = {Electron. J. Comb.},
	volume       = {13},
	number       = {1},
	year         = {2006},
	doi          = {10.37236/1139}
}

@article{DujmovicESW07,
	author       = {Vida Dujmovic and
	David Eppstein and
	Matthew Suderman and
	David R. Wood},
	title        = {Drawings of planar graphs with few slopes and segments},
	journal      = {Comput. Geom.},
	volume       = {38},
	number       = {3},
	pages        = {194--212},
	year         = {2007},
	doi          = {10.1016/J.COMGEO.2006.09.002}
}

@inproceedings{JelinekJKLTV09,
	author       = {V{\'{\i}}t Jel{\'{\i}}nek and
	Eva Jel{\'{\i}}nkov{\'{a}} and
	Jan Kratochv{\'{\i}}l and
	Bernard Lidick{\'{y}} and
	Marek Tesar and
	Tom{\'{a}}s Vyskocil},
	editor       = {David Eppstein and
	Emden R. Gansner},
	title        = {The Planar Slope Number of Planar Partial 3-Trees of Bounded Degree},
	booktitle    = {Graph Drawing and Network Visualization ({GD} 2009)},
	series       = {LNCS},
	volume       = {5849},
	pages        = {304--315},
	publisher    = {Springer},
	year         = {2009},
	doi          = {10.1007/978-3-642-11805-0_29}
}

@article{JelinekJKLTV13,
	author       = {V{\'{\i}}t Jel{\'{\i}}nek and
	Eva Jel{\'{\i}}nkov{\'{a}} and
	Jan Kratochv{\'{\i}}l and
	Bernard Lidick{\'{y}} and
	Marek Tesar and
	Tom{\'{a}}s Vyskocil},
	title        = {The Planar Slope Number of Planar Partial 3-Trees of Bounded Degree},
	journal      = {Graphs Comb.},
	volume       = {29},
	number       = {4},
	pages        = {981--1005},
	year         = {2013},
	doi          = {10.1007/S00373-012-1157-Z}
}

@article{KnauerMW14,
	author       = {Kolja B. Knauer and
	Piotr Micek and
	Bartosz Walczak},
	title        = {Outerplanar graph drawings with few slopes},
	journal      = {Comput. Geom.},
	volume       = {47},
	number       = {5},
	pages        = {614--624},
	year         = {2014},
	doi          = {10.1016/J.COMGEO.2014.01.003}
}

@inproceedings{KlawitterZ21,
	author       = {Jonathan Klawitter and
	Johannes Zink},
	editor       = {Helen C. Purchase and
	Ignaz Rutter},
	title        = {Upward Planar Drawings with Three and More Slopes},
	booktitle    = {Graph Drawing and Network Visualization ({GD} 2021)},
	series       = {LNCS},
	volume       = {12868},
	pages        = {149--165},
	publisher    = {Springer},
	year         = {2021},
	doi          = {10.1007/978-3-030-92931-2_11}
}

@article{KlawitterZ23,
	author       = {Jonathan Klawitter and
	Johannes Zink},
	title        = {Upward Planar Drawings with Three and More Slopes},
	journal      = {J. Graph Algorithms Appl.},
	volume       = {27},
	number       = {2},
	pages        = {49--70},
	year         = {2023},
	doi          = {10.7155/JGAA.00617}
}

@article{KlawitterM22,
	author       = {Jonathan Klawitter and
	Tamara Mchedlidze},
	title        = {Upward planar drawings with two slopes},
	journal      = {J. Graph Algorithms Appl.},
	volume       = {26},
	number       = {1},
	pages        = {171--198},
	year         = {2022},
	doi          = {10.7155/JGAA.00587}
}

@inproceedings{BekosGDLM18,
	author       = {Michael A. Bekos and
	Emilio {Di Giacomo} and
	Walter Didimo and
	Giuseppe Liotta and
	Fabrizio Montecchiani},
	editor       = {Therese Biedl and
	Andreas Kerren},
	title        = {Universal Slope Sets for Upward Planar Drawings},
	booktitle    = {Graph Drawing and Network Visualization ({GD} 2018)},
	series       = {LNCS},
	volume       = {11282},
	pages        = {77--91},
	publisher    = {Springer},
	year         = {2018},
	doi          = {10.1007/978-3-030-04414-5_6}
}

@article{BekosGDLM22,
	author       = {Michael A. Bekos and
	Emilio {Di Giacomo} and
	Walter Didimo and
	Giuseppe Liotta and
	Fabrizio Montecchiani},
	title        = {Universal Slope Sets for Upward Planar Drawings},
	journal      = {Algorithmica},
	volume       = {84},
	number       = {9},
	pages        = {2556--2580},
	year         = {2022},
	doi          = {10.1007/S00453-022-00975-3}
}

@inproceedings{BrucknerKM19,
	author       = {Guido Br{\"{u}}ckner and
	Nadine Davina Krisam and
	Tamara Mchedlidze},
	editor       = {Daniel Archambault and
	Csaba D. T{\'{o}}th},
	title        = {Level-Planar Drawings with Few Slopes},
	booktitle    = {Graph Drawing and Network Visualization ({GD} 2019)},
	series       = {LNCS},
	volume       = {11904},
	pages        = {559--572},
	publisher    = {Springer},
	year         = {2019},
	doi          = {10.1007/978-3-030-35802-0_42}
}

@article{BrucknerKM22,
	author       = {Guido Br{\"{u}}ckner and
	Nadine Davina Krisam and
	Tamara Mchedlidze},
	title        = {Level-Planar Drawings with Few Slopes},
	journal      = {Algorithmica},
	volume       = {84},
	number       = {1},
	pages        = {176--196},
	year         = {2022},
	doi          = {10.1007/S00453-021-00884-X}
}

@inproceedings{KnauerMW12,
  author       = {Kolja B. Knauer and
                  Piotr Micek and
                  Bartosz Walczak},
  editor       = {Joachim Gudmundsson and
                  Juli{\'{a}}n Mestre and
                  Taso Viglas},
  title        = {Outerplanar Graph Drawings with Few Slopes},
  booktitle    = {Computing and Combinatorics ({COCOON} 2012)},
  series       = {LNCS},
  volume       = {7434},
  pages        = {323--334},
  publisher    = {Springer},
  year         = {2012},
  doi          = {10.1007/978-3-642-32241-9_28},
}

@inproceedings{KnauerW16,
  author       = {Kolja  B. Knauer and
                  Bartosz Walczak},
  editor       = {Evangelos Kranakis and
                  Gonzalo Navarro and
                  Edgar Ch{\'{a}}vez},
  title        = {Graph Drawings with One Bend and Few Slopes},
  booktitle    = {Theoretical Informatics ({LATIN} 2016)},
  series       = {LNCS},
  volume       = {9644},
  pages        = {549--561},
  publisher    = {Springer},
  year         = {2016},
  doi          = {10.1007/978-3-662-49529-2_41}
}

@article{DidimoLM19,
  author       = {Walter Didimo and
                  Giuseppe Liotta and
                  Fabrizio Montecchiani},
  title        = {A Survey on Graph Drawing Beyond Planarity},
  journal      = {{ACM} Comput. Surv.},
  volume       = {52},
  number       = {1},
  pages        = {4:1--4:37},
  year         = {2019},
  url          = {https://doi.org/10.1145/3301281},
  doi          = {10.1145/3301281},
  bibsource    = {dblp computer science bibliography, https://dblp.org}
}

@article{KobourovLM17,
  author       = {Stephen G. Kobourov and
                  Giuseppe Liotta and
                  Fabrizio Montecchiani},
  title        = {An annotated bibliography on 1-planarity},
  journal      = {Comput. Sci. Rev.},
  volume       = {25},
  pages        = {49--67},
  year         = {2017},
  doi          = {10.1016/J.COSREV.2017.06.002},
}

@article{Ringel1965,
  title   = {Ein {S}echsfarbenproblem auf der {K}ugel},
  journal = {Abh. Math. Sem. Uni. Hamburg},
  volume  = {29},
  number  = {1},
  pages   = {107--117},
  year    = {1965},
  issn    = {1865-8784},
  doi     = {10.1007/BF02996313},
  author  = {Gerhard Ringel}
}

@inproceedings{BekosKKP26,
  author       = {Michael A. Bekos and
                  Eleni Katsanou and
                  Philipp Kindermann and
                  Maria Eleni Pavlidi},
  editor       = {Pierre Fraigniaud},
  title        = {How Many Slopes Does Polynomial Area Cost?},
  booktitle    = {Scandinavian Algorithm Theory ({SWAT} 2026)},
  series       = {LIPIcs},
  volume       = {370},
  pages        = {6:1--6:18},
  publisher    = {Schloss Dagstuhl - Leibniz-Zentrum f{\"{u}}r Informatik},
  year         = {2026},
  doi          = {10.4230/LIPICS.SWAT.2026.6},
}

@inproceedings{BattistaT90,
author = {Battista, Giuseppe Di and Tamassia, Roberto},
title = {On-Line Graph Algorithms with {SPQR}-Trees},
year = {1990},
publisher = {Springer},
series={LNCS},
booktitle = {Automata, Languages and Programming (ICALP 1990)},
pages = {598–611},
numpages = {14},
}

@inproceedings{ArgyriouCF0NORW18,
  author       = {Evmorfia N. Argyriou and
                  Sabine Cornelsen and
                  Henry F{\"{o}}rster and
                  Michael Kaufmann and
                  Martin N{\"{o}}llenburg and
                  Yoshio Okamoto and
                  Chrysanthi N. Raftopoulou and
                  Alexander Wolff},
  editor       = {Therese Biedl and
                  Andreas Kerren},
  title        = {Orthogonal and Smooth Orthogonal Layouts of 1-Planar Graphs with Low
                  Edge Complexity},
  booktitle    = {Graph Drawing and Network Visualization ({GD} 2018)},
  series       = {LNCS},
  pages        = {509--523},
  publisher    = {Springer},
  year         = {2018},
  doi          = {10.1007/978-3-030-04414-5_36},
}

\end{document}